\documentclass[onecolumn]{IEEEtran}
\usepackage[mathscr]{eucal}
\usepackage[cmex10]{amsmath}
\usepackage{epsfig,epsf,psfrag}
\usepackage{amssymb,amsmath,amsthm,amsfonts,latexsym}
\usepackage{amsmath,graphicx,bm,xcolor,url,overpic}
\usepackage{fixltx2e}
\usepackage{array}
\usepackage{verbatim}
\usepackage{bm}
\usepackage{algorithmic}
\usepackage{algorithm}
\usepackage{verbatim}
\usepackage{textcomp}
\usepackage{mathrsfs}
\usepackage{epstopdf}

\newcommand{\openone}{\leavevmode\hbox{\small1\normalsize\kern-.33em1}}

\catcode`~=11 \def\UrlSpecials{\do\~{\kern -.15em\lower .7ex\hbox{~}\kern .04em}} \catcode`~=13 

\allowdisplaybreaks[4]

\newcommand{\nn}{\nonumber}

\newcommand{\calA}{\mathcal{A}}
\newcommand{\calB}{\mathcal{B}}

\newcommand{\calD}{\mathcal{D}}

\newcommand{\calF}{\mathcal{F}}
\newcommand{\calG}{\mathcal{G}}

\newcommand{\calL}{\mathcal{L}}

\newcommand{\calN}{\mathcal{N}}

\newcommand{\calP}{\mathcal{P}}

\newcommand{\calR}{\mathcal{R}}

\newcommand{\calT}{\mathcal{T}}
\newcommand{\calU}{\mathcal{U}}
\newcommand{\calV}{\mathcal{V}}

\newcommand{\calX}{\mathcal{X}}
\newcommand{\calY}{\mathcal{Y}}
\newcommand{\calZ}{\mathcal{Z}}

\newcommand{\bV}{\mathbf{V}}

\newcommand{\bx}{\mathbf{x}}

\newcommand{\rmc}{\mathrm{c}}

\newcommand{\rmd}{\mathrm{d}}

\newcommand{\rmQ}{\mathrm{Q}}

\newcommand{\rmV}{\mathrm{V}}

\newcommand{\bbR}{\mathbb{R}}

\DeclareMathAlphabet{\mathbsf}{OT1}{cmss}{bx}{n}
\DeclareMathAlphabet{\mathssf}{OT1}{cmss}{m}{sl}

\newcommand{\rvR}{\mathsf{R}}

\DeclareSymbolFont{bsfletters}{OT1}{cmss}{bx}{n}  
\DeclareSymbolFont{ssfletters}{OT1}{cmss}{m}{n}
\DeclareMathSymbol{\bsfGamma}{0}{bsfletters}{'000}
\DeclareMathSymbol{\ssfGamma}{0}{ssfletters}{'000}
\DeclareMathSymbol{\bsfDelta}{0}{bsfletters}{'001}
\DeclareMathSymbol{\ssfDelta}{0}{ssfletters}{'001}
\DeclareMathSymbol{\bsfTheta}{0}{bsfletters}{'002}
\DeclareMathSymbol{\ssfTheta}{0}{ssfletters}{'002}
\DeclareMathSymbol{\bsfLambda}{0}{bsfletters}{'003}
\DeclareMathSymbol{\ssfLambda}{0}{ssfletters}{'003}
\DeclareMathSymbol{\bsfXi}{0}{bsfletters}{'004}
\DeclareMathSymbol{\ssfXi}{0}{ssfletters}{'004}
\DeclareMathSymbol{\bsfPi}{0}{bsfletters}{'005}
\DeclareMathSymbol{\ssfPi}{0}{ssfletters}{'005}
\DeclareMathSymbol{\bsfSigma}{0}{bsfletters}{'006}
\DeclareMathSymbol{\ssfSigma}{0}{ssfletters}{'006}
\DeclareMathSymbol{\bsfUpsilon}{0}{bsfletters}{'007}
\DeclareMathSymbol{\ssfUpsilon}{0}{ssfletters}{'007}
\DeclareMathSymbol{\bsfPhi}{0}{bsfletters}{'010}
\DeclareMathSymbol{\ssfPhi}{0}{ssfletters}{'010}
\DeclareMathSymbol{\bsfPsi}{0}{bsfletters}{'011}
\DeclareMathSymbol{\ssfPsi}{0}{ssfletters}{'011}
\DeclareMathSymbol{\bsfOmega}{0}{bsfletters}{'012}
\DeclareMathSymbol{\ssfOmega}{0}{ssfletters}{'012}

\newcommand{\hatT}{\hat{T}}

\newcommand{\bmu}{\bm{\mu}}

\newcommand{\bSigma	}{\bm{\Sigma}}

\newcommand{\ceil}[1]{\lceil{#1}\rceil}

\DeclareMathOperator*{\argmin}{arg\,min}

\newtheorem{theorem}{Theorem} 
\newtheorem{lemma}[theorem]{Lemma}

\newtheorem{corollary}[theorem]{Corollary}
\newtheorem{definition}{Definition}

\graphicspath{{./figures/}}
\usepackage{epstopdf}
\usepackage[ colorlinks = true,
             linkcolor = blue,
             urlcolor  = blue,
             citecolor = red,
             anchorcolor = green,]{hyperref}

\begin{document}
\flushbottom
\title{Second-Order and Moderate Deviations Asymptotics for Successive Refinement}
\author{Lin Zhou  $\qquad$Vincent Y.~F.~Tan$\qquad$
        Mehul Motani
\thanks{The authors are with the  Department of Electrical and Computer Engineering, National University of Singapore (NUS). V.~Y.~F.~Tan is also with the Department of Mathematics, NUS. Emails: \url{lzhou@u.nus.edu}; \url{vtan@nus.edu.sg}; \url{motani@nus.edu.sg}.}
\thanks{Part of this paper has been presented in~\cite{zhou2016sr} at ISIT 2016, Barcelona, Spain.}
}
\maketitle

\begin{abstract}
We derive the optimal second-order coding region and moderate deviations constant for successive refinement source coding with a joint excess-distortion probability constraint.  We consider two scenarios: (i) a discrete memoryless source (DMS) and arbitrary distortion measures at the decoders and (ii) a Gaussian memoryless source (GMS) and quadratic distortion measures at the decoders. 
For a DMS with arbitrary distortion measures, we prove an achievable second-order coding region, using type covering lemmas by Kanlis and Narayan and by No, Ingber and Weissman.  We prove the converse using the perturbation approach by Gu and Effros.  
When the DMS is successively refinable, the expressions for the second-order coding region and the moderate deviations constant are simplified and easily computable.  For this case, we also obtain new insights on the second-order behavior compared to the scenario where separate excess-distortion proabilities are considered.  For example, we describe a DMS, for which the optimal second-order region transitions from being characterizable by a bivariate Gaussian to a univariate Gaussian, as the distortion levels are varied.
We then consider a GMS with quadratic distortion measures.  To prove the direct  part, we make use of the sphere covering theorem by Verger-Gaugry, together with appropriately-defined Gaussian type classes.  To prove the converse, we generalize Kostina and Verd\'u's one-shot converse bound for point-to-point lossy source coding.  We remark that this proof is applicable to general successively refinable sources.  In the proofs of the moderate deviations results for both scenarios, we follow a strategy similar to that for the second-order asymptotics and use the moderate deviations principle. 
\end{abstract}

\begin{IEEEkeywords}
Successive refinement, Second-order asymptotics, Moderate deviations, Discrete memoryless source, Gaussian memoryless source, Gaussian types
\end{IEEEkeywords}

\section{Introduction}

The successive refinement source coding problem~\cite{rimoldi1994,equitz1991successive} is shown in Figure~\ref{systemmodel}. There are two encoders and two decoders. Encoder $f_i,~i=1,2$ has access to a source sequence $X^n$ and compresses it into a message $S_i,~i=1,2$. Decoder $\phi_1$ aims to recover source sequence $X^n$ under distortion measure $d_1$ and distortion level $D_1$ with the encoded message $S_1$ from encoder $f_1$. The decoder $\phi_2$ aims to recover $X^n$ under distortion measure $d_2$ and distortion level $D_2$ with messages $S_1$ and $S_2$. The optimal rate region for a DMS with arbitrary distortion measures was characterized by Rimoldi in~\cite{rimoldi1994}.  This problem has many practical applications in image and video compression. For example, we may want to describe an image optimally to within  a particular amount of distortion; later when we obtain more information about the image, we hope to specify it more accurately. The successive refinement problem is an information-theoretic formulation of whether its is possible to interrupt a transmission at any time without any loss of optimality in compression~\cite{rimoldi1994}.

In this paper, we analyze two    asymptotic regimes associated with the successive refinement problem---namely, the second-order~\cite{tan2015asymptotic} and the moderate deviations asymptotic regimes~\cite{altugwagner2014}. Our analysis provides a more refined picture on the performance  of optimal codes for the setting in which the {\em joint excess-distortion probability} (in contrast to the separate  excess-distortion probabilities in \cite{no2015strong}) is non-vanishing and the setting in which this probability decays sub-exponentially fast. By joint excess-distortion probability, we mean the probability that {\em either} of the two decoders fails to reproduce  the source $X^n$ to within prescribed distortion levels $D_1$ or $D_2$. In contrast, the separate excess-distortion probability formalism places  separate upper bounds on each of the probabilities that the source is not reproduced to within distortion levels $D_1$ and $D_2$. Let us now explain some advantages of using the joint  criterion over the separate one.
\begin{enumerate}
\item The joint criterion is consistent with recent works in the second-order literature~\cite{tan2015asymptotic,vincent2014dispersion,le2015inter,watanabe2015second}. For example, in \cite{le2015inter}, Le, Tan and Motani established the second-order asymptotics for the Gaussian interference channel in the strictly very strong interference regime under the joint error probability criterion. If in \cite{le2015inter}, one adopts the separate error probabilities criterion, one would {\em not} be    able to observe the performance tradeoff between the two decoders. 
\item In Section \ref{sec:examples},  we show, via different proof techniques compared to existing works, that the second-order region (when the rate of a code is located at a corner point of the first-order rate region) is curved. This shows that if one second-order coding rate is small, the other is necessarily large. This reveals a fundamental tradeoff that cannot be observed if one adopts the separate excess-distortion probability criterion.
\item In moderate deviations analysis (see case (iii) in Theorem \ref{mdconstant} and Corollary \ref{srmdc}), under the joint criterion, we observe that the worse decoder dominates the overall performance. This parallels error exponent analysis of Kanlis and Narayan~\cite{kanlis1996error} and can only be observed under the joint excess-distortion probability criterion. 
\end{enumerate}

In this work, 
we study two classes of sources, namely discrete and Gaussian memoryless sources.

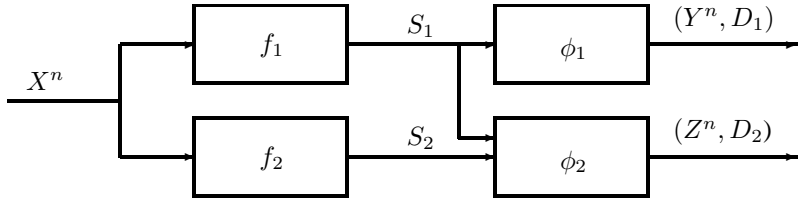
\begin{figure}[t]
\centering
\setlength{\unitlength}{0.5cm}
\scalebox{1}{
\begin{picture}(26,6)
\linethickness{1pt}
\put(1.5,3.8){\makebox{$X^n$}}
\put(6,1){\framebox(4,2)}
\put(6,4){\framebox(4,2)}
\put(7.7,1.8){\makebox{$f_2$}}
\put(7.7,4.8){\makebox{$f_1$}}
\put(1,3.5){\line(1,0){3}}
\put(4,5){\vector(1,0){2}}
\put(4,2){\line(0,1){3}}
\put(4,2){\vector(1,0){2}}
\put(14,1){\framebox(4,2)}
\put(14,4){\framebox(4,2)}

\put(15.7,4.7){\makebox{$\phi_1$}}
\put(15.7,1.7){\makebox{$\phi_2$}}
\put(10,2){\vector(1,0){4}}
\put(12,2.5){\makebox(0,0){$S_2$}}
\put(10,5){\vector(1,0){4}}
\put(12,5.5){\makebox(0,0){$S_1$}}
\put(13,5){\line(0,-1){2.5}}
\put(13,2.5){\vector(1,0){1}}
\put(18,2){\vector(1,0){4}}
\put(18.7,2.5){\makebox{$(Z^n,D_2$)}}
\put(18,5){\vector(1,0){4}}
\put(18.7,5.5){\makebox{$(Y^n,D_1)$}}
\end{picture}}
\caption{The Successive Refinement Source Coding Problem~\cite{rimoldi1994}.}
\label{systemmodel}
\end{figure}

\subsection{Main Contributions}
There are two main contributions in this paper. 

First, for a DMS with arbitrary distortion measures, we derive the optimal second-order coding region and moderate deviations constant for the successive refinement source coding problem under a {\em joint} excess-distortion criterion in contrast to the {\em separate} excess-distortion criteria in No, Ingber and Weissman~\cite{no2015strong}. As mentioned above, we opine that the joint criterion is also important and is, in fact, in line with  the original work by Rimoldi~\cite{rimoldi1994} and the work on error exponents (the reliability function) by Kanlis and Narayan~\cite{kanlis1996error}. There are several new insights on the second-order coding region that we can glean when we consider the joint excess-distortion probability (cf.\ Section \ref{sec:examples}). Moreover, we show that our result can be specialized to successively refinable discrete memoryless source-distortion measure triplets, leading to a simpler second-order region and also a simpler expression for the moderate deviations constant. In the achievability part, we leverage the type covering lemma~\cite[Lemma 8]{no2015strong}. In the converse part, we follow the perturbation approach by Gu and Effros in their proof for the strong converse of Gray-Wyner problem~\cite{wei2009strong}, leading to a type-based strong converse. In the proofs of both directions, we leverage the properties of appropriately-defined distortion-tilted information densities and we also use the (multi-variate) Berry-Esseen theorem~\cite{Ben03} and the moderate deviations principle/theorem in~\cite[Theorem~3.7.1]{dembo2009large}. Furthermore, in the proof of converse part for successively refinable source-distortion measure triplets, we generalize the one-shot converse bound of Kostina and Verd\'u in~\cite[Theorem 1]{kostina2012converse}. We remark that this converse proof is also applicable to successively refinable continuous memoryless source-distortion measure triplets such as the a GMS with quadratic distortion measures. For the moderate deviations analysis for a DMS with arbitrary distortion measures, we use an information spectrum calculation similar to that used  for the second-order asymptotics analysis.

Our second contribution pertains to a GMS with quadratic distortion measures in which we establish the second-order coding region and the moderate deviations constant. The solutions are particularly simple because a GMS with quadratic distortion measures is successively refinable~\cite{equitz1991successive}. However, because the Gaussian source is continuous, we need to modify the type covering lemma mentioned above, as it only applies to discrete sources there. We apply the sphere covering theorem~\cite{verger2005covering} multiple times to establish a Gaussian type covering lemma for the successive refinement problem. To subsequently apply this lemma to calculate the joint excess-distortion probability, we need to define the notion of Gaussian types (cf.\ \cite{arikan1998guessing,kelly2012reliability}) carefully. Indeed, the quantizations of the power of the source for the second-order and moderate deviations analyses are different and they need to be chosen carefully. We note that appropriately-defined Gaussian types have been used in the work by Scarlett for the second-order asymptotics of the dirty-paper problem~\cite{scarlett2015} and Scarlett and Tan's work for the second-order asymptotics of the Gaussian MAC with degraded message sets~\cite{scarlett2015a}.

\subsection{Related Work}
We briefly summarize other works that are related to successive refinement source coding. ~\cite{effros1999} extended Rimoldi's result in~\cite{rimoldi1994} to discrete stationary ergodic and non-ergodic sources. Motivated by memory limitation concerns, Tuncel and Rose considered additive successive refinement source coding problems in~\cite{tuncel2003additive} where the decoding scheme is   constrained to be additive over an Abelian (commutative) group.  Kanlis and Narayan~\cite{kanlis1996error} derived the error exponent under the  joint excess-distortion criterion while Tuncel and Rose~\cite{tuncel2003} considered the separate excess-distortion criterion for  two layers. Second-order coding rates were derived for the so-called {\em strong} successive refinement problem  by No, Ingber and Weissman~\cite{no2015strong}. They considered the {\em separate} excess-distortion criteria. In this work, we consider the {\em joint} excess-distortion criterion.

There are several works that consider second-order asymptotics for lossless and lossy source coding. Strassen~\cite{strassen1962asymptotische} derived the second-order coding rate for point-to-point lossless source coding and Hayashi~\cite{hayashi2008source} revisited the problem using information spectrum method. Tan and Kosut~\cite{tan2014dispersions} and Nomura and Han~\cite{nomura2014} considered the Slepian-Wolf problem and Watanabe considered  the lossless Gray-Wyner problem~\cite{watanabe2015second}. The dispersion for point-to-point lossy source coding was derived by Ingber and Kochman~\cite{ingber2011} and by Kostina and Verd\'u~\cite{kostina2012fixed}. The Wyner-Ziv problem was considered by Watanabe, Kuzuoka and Tan in~\cite{watanabe2015} and by Yassaee, Aref and Gohari in~\cite{yassaee2013technique}. In a work that can be considered dual to source coding, Kumagai and Hayashi~\cite{kumagai2014} studied the second-order asymptotics of random number conversion in quantum information and noticed that interestingly, the asymptotic distribution of interest is not the usual normal distribution but a  generalized Rayleigh-normal distribution. 

We also recall the  related works on moderate deviations analysis. Chen  {\em et al.}~\cite{chen2007redundancy} and He {\em et al.}~\cite{he2009redundancy} initiated the study of moderate deviations for  fixed-to-variable length source coding with decoder side information. For fixed-to-fixed length analysis, Altu\u{g} and Wagner~\cite{altugwagner2014} initiated the study of moderate deviations in the context of discrete memoryless channels. Polyanksiy and Verd\'u~\cite{polyanskiy2010channel} relaxed some assumptions in the conference version of Altu\u{g} and Wagner's work~\cite{altug2010moderate} and they also considered moderate deviations for AWGN channels. Altu\u{g}, Wagner and Kontoyiannis~\cite{altug2013lossless} considered moderate deviations for lossless source coding. For lossy source coding, the moderate deviations analysis was done by Tan in~\cite{tan2012moderate} using ideas from Euclidean information theory~\cite{borade2008}.

\subsection{Organization of the Paper}
The rest of the paper is organized as follows. In Section~\ref{sec:pf}, we set up the notation, formulate the successive refinement source coding problem and recall existing results including the first-order rate region and conditions for a source-distortion measure triplet to be successively refinable. In Section~\ref{sec:maindms}, we present the second-order coding region and moderate deviations constant for a DMS with arbitrary distortion measures and specialize the result to successively refinable discrete memoryless source-distortion measure triplets. We illustrate our results using two examples from Kostina and Verd\'u \cite{kostina2012fixed}, leading to new insights on the second-order fundamental limits. Furthermore, we generalize the one-shot lower bound by Kostina and Verd\'u in~\cite[Theorem 1]{kostina2012converse} to provide an alternative converse proof for successively refinable source-distortion measure triplets. Respectively in Sections~\ref{secondproof} and~\ref{proofmdc}, we present the proofs for the second-order asymptotics and moderate deviations results for a DMS. In Section~\ref{sec:maingms}, we present the second-order coding region and moderate deviations constant together with their proofs for a GMS with quadratic distortion measures. Finally, in Section~\ref{sec:conc}, we conclude the paper. To present the main results of the paper seamlessly, we defer the proofs of all supporting technical lemmas to the appendices.

\section{Problem Formulation and Existing Results}
\label{sec:pf}
\subsection{Notation}
Random variables and their realizations are in capital (e.g.,\ $X$) and lower case (e.g.,\ $x$) respectively. All sets are denoted in calligraphic font (e.g.,\ $\mathcal{X}$). We use $\calX^{\mathrm{c}}$ to denote the complement of $\calX$. Let $X^n:=(X_1,\ldots,X_n)$ be a random vector of length $n$. We use $\|x^n\| =\sqrt{\sum_i x_i^2}$ to denote the $l_2$ norm of the vector $x^n\in\bbR^n$. We use $\exp(x)$ to denote $e^x$. All logarithms are base $e$ (except in Section \ref{sec:examples} where we use base $2$). We use $\rmQ(\cdot)$ to denote the standard Gaussian complementary cumulative distribution function (cdf) and  $\rmQ^{-1}(\cdot )$ its inverse. Given two integers $a$ and $b$, we use $[a:b]$ to denote all the integers between $a$ and $b$. We use standard asymptotic notation such as $O(\cdot)$  and $o(\cdot)$. We use $\mathbb{R}_+$ to denote the set of non-negative real numbers and $\mathrm{ones}(m_1, m_2)$ to denote the $m_1\times m_2$ matrix of all ones. For mutual information, we use $I(X;Y)$ and $I(P_{X},P_{Y|X})$ interchangeably.

The set of all probability distributions on a finite set $\calX$ is denoted as $\calP(\calX)$ and the set of all conditional probability distributions from $\calX$ to $\calY$ is denoted as $\calP(\calY|\calX)$. Given $P\in\calP(\calX)$ and $V\in\calP(\calY|\calX)$, we use $P\times V$ to denote the joint distribution induced by $P$ and $V$. In terms of the method of types for a DMS, we use the notation as~\cite{tan2015asymptotic}. Given sequence $x^n$, the empirical distribution is denoted as $\hat{T}_{x^n}$. The set of types formed from length $n$ sequences in $\calX$ is denoted as $\calP_{n}(\calX)$. Given $P\in\calP_{n}(\calX)$, the set of all sequences of length $n$ with type $P$ is denoted as $\calT_{P}$.

\subsection{Problem Formulation}
We consider a memoryless source with distribution $P_{X}$ supported on an arbitrary (discrete or continuous) alphabet $\calX$. Hence $X^n$ is an i.i.d. sequence where each $X_i$ is generated according to $P_{X}$. We assume the reproduction alphabets for decoder $\phi_1,\phi_2$ are respectively alphabets $\calY$ and $\calZ$. We follow the definitions in~\cite{rimoldi1994} for codes and achievable rate region.
\begin{definition}
An $(n,M_1,M_2)$-code for successive refinement source coding consists of two encoders:
\begin{align}
f_1:\calX^n\to\{1,2,\ldots,M_1\},\\
f_2:\calX^n\to\{1,2,\ldots,M_2\},
\end{align}
and two decoders:
\begin{align}
\phi_1&:\{1,2,\ldots,M_1\}\to \calY^n,\\*
\phi_2&:\{1,2,\ldots,M_1\}\times\{1,2,\ldots,M_2\}\to \calZ^n.
\end{align}
\end{definition}

Define two distortion measures: $d_1:\calX\times\calY\to[0,\infty)$ and $d_2:\calX\times\calZ\to[0,\infty)$ such that for each $x\in\calX$, there exists $y\in\calY$ and $z\in\calZ$ satisfying $d_1(x,y)=0$ and $d_2(x,z)=0$. Let the distortion between $x^n$ and $y^n$ be defined as $d_1(x^n,y^n):=\frac{1}{n}\sum_{i=1}^nd_1(x_i,y_i)$ and the distortion $d_2(x^n,z^n)$ be defined in a similar manner. Throughout the paper, we consider the case where $D_1>0$ and $D_2>0$. Define the {\em joint excess-distortion probability} as
\begin{align}
\label{defexcessprob}
\epsilon_n(D_1,D_2)&:=\Pr\left(d_1(X^n,Y^n)>D_1~\mathrm{or}~d_2(X^n,Z^n)> D_2\right),
\end{align}
where $Y^n=\phi_1(f_1(X^n))$ and $Z^n = \phi_2(f_1(X^n), f_2(X^n))$ are the reconstructed sequences.
\begin{definition}[First-order Region]
\label{deffirst}
A rate pair $(R_1,R_2)$ is said to be $(D_1,D_2)$-achievable for the successive refinement source coding if there exists a sequence of $(n,M_1,M_2)$-codes such that
\begin{align}
\limsup_{n\to\infty}\frac{1}{n}\log M_1\leq R_1,\\*
\limsup_{n\to\infty}\frac{1}{n}\log(M_1M_2)\leq R_2\label{eqn:R2},
\end{align}
and
\begin{align}
\lim_{n\to\infty} \epsilon_n(D_1,D_2)=0.
\end{align}
The closure of the set of all $(D_1,D_2)$-achievable rate pairs is called optimal  $(D_1,D_2)$-achievable rate region and denoted as $\calR(D_1,D_2|P_{X})$.
\end{definition}
Note from~\eqref{eqn:R2} that $R_2$ corresponds to an upper bound on the {\em sum} rate (and not the rate of message $S_2$ in Figure~\ref{systemmodel}). This is in line with the original work by Rimoldi~\cite{rimoldi1994}.

Now for the following two definitions, we set $(R_1^*, R_2^*)$ to be a  rate pair  on the boundary of $\calR(D_1,D_2|P_{X})$.
\begin{definition}[Second-order Region]
\label{defsecond}
A pair $(L_1,L_2)$ is said to be second-order $(R_1^*,R_2^*,D_1,D_2,\epsilon)$-achievable if there exists a sequence of $(n,M_1,M_2)$-codes such that
\begin{align}
\limsup_{n\to\infty}\frac{1}{\sqrt{n}}\left(\log M_1-nR_1^*\right)\leq L_1,\\
\limsup_{n\to\infty}\frac{1}{\sqrt{n}}\left(\log(M_1M_2)-nR_2^*\right)\leq L_2,
\end{align}
and
\begin{align}
\limsup_{n\to\infty}\epsilon_n(D_1,D_2)\leq \epsilon.
\end{align}
The closure of the set of all second-order $(R_1^*,R_2^*,D_1,D_2,\epsilon)$-achievable pairs is called the optimal second-order $(R_1^*,R_2^*,D_1,D_2,\epsilon)$-achievable coding region and denoted as $\calL(R_1^*,R_2^*,D_1,D_2,\epsilon)$.
\end{definition}

We emphasize that we consider the {\em joint} excess-distortion probability~\eqref{defexcessprob} which is consistent with original setting in Rimoldi's work~\cite{rimoldi1994} and the work on error exponents by Kanlis and Narayan~\cite{kanlis1996error}. This is in contrast to the work by  No, Ingber and Weissman who considered separate excess-distortion events and probabilities~\cite[Definition 3]{no2015strong}. That is, they considered the setting in which   the code satisfies
\begin{align}
\epsilon_{1,n}(D_1):=\Pr\left(d_1(X^n,Y^n)>D_1\right)&\le\eta_1,\quad\mbox{and} \label{eqn:eta1}\\*
\epsilon_{2,n}(D_2):=\Pr\left(d_2(X^n,Z^n)>D_2\right)&\le\eta_2\label{eqn:eta2}
\end{align}
for some fixed $(\eta_1,\eta_2)\in (0,1)^2$. We opine that the analysis of the probability of the {\em joint excess-distortion event} in~\eqref{defexcessprob} is also of significant interest. We remark that the first-order fundamental limit (rate region) remains the same~\cite{rimoldi1994,equitz1991successive} regardless whether we consider the joint or the separate excess-distortion probabilities.     However, under joint criterion,  we are able  to obtain new insights about the second-order fundamental limits of the successive refinement problem as can be seen from the example  in Subsection \ref{sec:quat}.

\begin{definition}[Moderate Deviations Constant]
\label{defmdconstant}
Consider any sequence $\{\rho_n\}_{n=1}^{\infty}$ satisfying
\begin{align}
\label{epsilonn}
 \lim_{n\to\infty}\rho_n&=0,\\
 \lim_{n\to\infty}\sqrt{n}\rho_n&=\infty\label{normalmdc}.
\end{align}
Let $\theta_i,~i=1,2$ be two fixed positive real numbers.
A number $\nu$ is said to be a $(R_1^*,R_2^*)$-achievable moderate deviations constant if there exists a sequence of $(n,M_1,M_2)$-codes such that
\begin{align}
\limsup_{n\to\infty}\frac{1}{n\rho_n}(\log M_1-nR_1^*)&\leq \theta_1, \label{eqn:theta_1}\\
\limsup_{n\to\infty}\frac{1}{n\rho_n}(\log(M_1M_2)-nR_2^*)&\leq \theta_2,\label{eqn:theta_2}
\end{align}
and  
\begin{align}
 \liminf_{n\to\infty} -\frac{\log \epsilon_n(D_1,D_2)}{n\rho_n^2}\geq \nu\label{jmdcerror}.
\end{align}
The supremum of all $(R_1^*,R_2^*)$-achievable moderate deviations constants is denoted as $\nu^*(R_1^*,R_2^*|D_1,D_2)$. 
\end{definition}
We remark that the constants $\theta_1$ and $\theta_2$ are present in \eqref{eqn:theta_1} and \eqref{eqn:theta_2} to reflect possibly different speeds of convergence of $\frac{1}{n}\log M_1$ and $\frac{1}{n}\log (M_1 M_2)$ to $R_1^*$ and $R_2^*$ respectively. The speeds are $O(\rho_n)$ but the constants in this $O(\cdot)$  notation  are different.

The central goal of this paper is to characterize $\calL(R_1^*,R_2^*,D_1,D_2,\epsilon)$ and $\nu^*(R_1^*,R_2^*|D_1,D_2)$ for a DMS with arbitrary distortion measures (e.g., a binary source with Hamming distortion measures) and a GMS with quadratic distortion measures. We note that $\calL(R_1,R_2 ,D_1,D_2,\epsilon)$ and $\nu^*(R_1 ,R_2 |D_1,D_2)$  can, in principle, be evaluated for rate pairs that are not on the boundary of the first-order region  $\calR(D_1,D_2|P_{X})$. However, this would lead to degenerate solutions by the achievability of all rate pairs in the interior of  $\calR(D_1,D_2|P_{X})$ and the strong converse  for all rate pairs in the exterior of  $\calR(D_1,D_2|P_{X})$, a direct corollary of our main result in Theorem \ref{mainresult}. Note that the strong converse for the successive refinement problem was originally  established by Rimoldi~\cite{rimoldi1994}.

\subsection{Existing Results}
The optimal rate region for a DMS with arbitrary distortion measures was characterized in~\cite{rimoldi1994}. Let $\calP(P_X,D_1,D_2)$ be the set of joint distributions $P_{XYZ}$ such that the $\calX$-marginal is $P_{X}$, $\mathbb{E}[d_1(X,Y)]\leq D_1$ and $\mathbb{E}[d_{2}(X,Z)]\leq D_2$.
\begin{theorem}
\label{rimoldi}
The optimal $(D_1,D_2)$-achievable rate region for a DMS with arbitrary distortion measures under successive refinement source coding is
\begin{align}
\calR(D_1,D_2|P_{X})=\bigcup_{P_{XYZ}\in\calP(P_X,D_1,D_2)}\left\{(R_1,R_2):R_1\geq I(X;Y),~R_2\geq I(X;YZ)\right\}\label{eqn:opt_rr}.
\end{align}
\end{theorem}

Now we introduce an important quantity for subsequent analyses for a DMS. Given a rate $R_1$ and distortion pair $(D_1,D_2)$, let the minimum rate $R_2$ such that $(R_1,R_2)\in\calR(D_1,D_2|P_{X})$ be $\rvR(R_1,D_1,D_2|P_X)$, i.e.,
\begin{align}
\label{minimumr2}
\rvR(R_1,D_1,D_2|P_X)
&:=\min\left\{R_2:(R_1,R_2)\in\calR(D_1,D_2|P_{X})\right\}\\*
&=\inf_{\substack{P_{YZ|X}:\mathbb{E}[d_1(X,Y)]\leq D_1\\ \mathbb{E}[d_2(X,Z)]\leq D_2, I(X;Y)\leq R_1}} I(X;YZ)\label{minr2},
\end{align}
where~\eqref{minr2} follows from~\cite[Corollary 1]{kanlis1996error}.

Note that if $R_1<R_Y(P_X,D_1)$, then the convex optimization in~\eqref{minr2} is infeasible, hence $\rvR(R_1,D_1,D_2|P_X)=\infty$. For other cases, since $\rvR(R_1,D_1,D_2|P_X)$ is a convex optimization problem, the minimization in~\eqref{minr2} is attained for some test channel $P_{YZ|X}$ satisfying  
\begin{align}
\sum_{x,y,z}P_{X}(x)P_{YZ|X}(yz|x)d_1(x,y)&=D_1,\\
\sum_{x,y,z}P_{X}(x)P_{YZ|X}(yz|x)d_2(x,z)&=D_2,\\
I(P_X,P_{Y|X})&=R_1.
\end{align}

Now we introduce the notion of a {\em successively refinable source-distortion measure triplet}~\cite{koshelev1981estimation,equitz1991successive}. For such a source-distortion measure triplet, the minimum $R_2$ given $R_1$ in a certain interval is exactly the rate-distortion function (see~\eqref{srminr2} to follow). This reduces the computation of the optimal rate region in~\eqref{eqn:opt_rr}. We recall the definitions with a slight generalization in accordance to~\cite[Definition~2]{no2015strong}.
Let $R_Y(P_X,D_1)$ and $R_Z(P_X,D_2)$ be  the rate-distortion functions~\cite[Chapter~3]{el2011network} when the reproduction alphabets are $\calY$ and $\calZ$ respectively, i.e.,
\begin{align}
R_Y(P_X,D_1) &:=\inf_{P_{Y|X}:\mathbb{E}[d_1(X,Y)]\leq D_1} I(X;Y),\\
R_Z(P_X,D_2)&:=\inf_{P_{Z|X}:\mathbb{E}[d_2(X,Z)]\leq D_2} I(X;Z).
\end{align}

\begin{definition}
Given distortion measures $d_1,d_2$ and a source $X$ with distribution $P_X$. A source-distortion measure triplet $(X,d_1,d_2)$ is said to be $(D_1,D_2)$-successively refinable if $(R_Y(P_X,D_1),R_Z(P_X,D_2))\in\calR(D_1,D_2|P_X)$. If the source-distortion measure triplet is $(D_1,D_2)$-successively refinable for all $(D_1,D_2)$ such that $R_Y(P_X,D_1)<R_Z(P_X,D_2)$, then it  is said to be successively refinable.
\end{definition}
Koshelev~\cite{koshelev1981estimation} presented a sufficient condition for a source-distortion measure triplet to be successively refinable while Equitz and Cover~\cite[Theorem 2]{equitz1991successive} presented a  
necessary and sufficient condition which we reproduce below. 

\begin{theorem}
A memoryless source-distortion measure triplet is successively refinable if and only if there exists a conditional distribution $P_{YZ|X}^*$ such that
\begin{align}
R_Y(P_X,D_1)=I(P_X,P_{Y|X}^*),~\mathbb{E}_{P_X\times P_{Y|X}^*}[d_1(X,Y)]\leq D_1,\\
R_Z(P_X,D_2)=I(P_X,P_{Z|X}^*),~\mathbb{E}_{P_X\times P_{Z|X}^*}[d_2(X,Z)]\leq D_2,
\end{align}
and
\begin{align}
P_{YZ|X}^*=P_{Y|Z}^*P_{Z|X}^*.
\end{align}
\end{theorem}
In~\cite{equitz1991successive}, it was shown that a DMS with Hamming distortion measures, a GMS with quadratic distortion measures and a Laplacian source with absolute distortion measures are successively refinable. Note that in the original paper of Equitz and Cover~\cite{equitz1991successive}, the authors only considered $d_1=d_2=d$. However, as pointed out in~\cite[Theorem 4]{no2015strong}, the result holds even when $d_1\neq d_2$. This can be verified easily for a DMS by invoking~\cite[Theorem 1]{rimoldi1994}.

For a successively refinable discrete memoryless source-distortion measure triplet, it is obvious that
\begin{align}
\rvR(R_Y(P_X,D_1),D_1,D_2|P_X)
&=R_Z(P_X,D_2).
\end{align}
Recall that $\rvR(R_1,D_1,D_2|P_X)$ is a non-increasing function of $R_1$. Hence for $R_Y(P_X,D_1)\leq R_1< R_Z(P_X,D_2)$,
\begin{align}
\rvR(R_1,D_1,D_2|P_X)=R_Z(P_X,D_2)\label{srminr2}.
\end{align}
We then recall the definition of distortion-tilted information density~\cite[Definition 1]{kostina2012converse}. Let $P_{Y}^*$ be induced by $P_{Y|X}^*$ which achieves $R_Y(P_X,D_1)$ and $P_Z^*$ be induced by $P_{Z|X}^*$ which achieves   $R_Z(P_X,D_2)$. The $D_1$-tilted information density~\cite{kostina2012converse} is defined as follows:
\begin{align}
\label{kostinatilt}
\jmath_Y(x,D_1|P_X):=-\log \mathbb{E}_{P_{Y}^*}[\exp(-s_1^*(d(x,Y)-D_1))],
\end{align}
where
\begin{align}
s_1^*=-\frac{\partial R_Y(P_X,D)}{\partial D}\bigg|_{D=D_1},
\end{align}
while $\jmath_Z(x,D_2|P_X)$ and $s_2^*$ are defined similarly. The properties of $\jmath_Y(x,D_1|P_X)$ and $\jmath_Z(x,D_2|P_X)$ were derived in~\cite[Properties 1-3]{kostina2012converse} and~\cite[Theorems 2.1 \& 2.2]{kostina2013lossy}

\section{A Discrete Memoryless Source with Arbitrary Distortion Measures}
\label{sec:maindms}
In this section, we consider a DMS in which the alphabets $\calX$, $\calY$ and $\calZ$ are all \emph{finite}.

\subsection{Tilted Information Density}

Throughout the section, we assume that $R_Y(P_X,D_1)\leq R_1^*<R_Z(P_X,D_2)$ and $\calR(D_1,D_2|P_X)$ is smooth on a boundary rate pair $(R_1^*,R_2^*)$ of our interest, i.e.,
\begin{align}
\label{deflambda}
\lambda^*&:=-\frac{\rvR(R,D_1,D_2|P_X)}{\partial R}\bigg|_{R=R_1^*}, 
\end{align}
is well-defined. Note that $\lambda^*\geq 0$ since $\rvR(R_1,D_2,D_2)$ is a convex and non-increasing function in $R_1$. Further, for  a positive distortion pair $(D_1,D_2)$, define
\begin{align}
\nu_1^*&:=-\frac{\rvR(P_X,R_1,D,D_2)}{\partial D}\bigg|_{D=D_1}\label{defnu1},\\*
\nu_2^*&:=-\frac{\rvR(P_X,R_1,D_1,D)}{\partial D}\bigg|_{D=D_2}\label{defnu2}.
\end{align}
Note that for a successively refinable discrete memoryless source-distortion measure triplet, from~\eqref{srminr2}, we obtain $\lambda^*=0$ and $\nu_1^*=0$. Let $P_{YZ|X}^*$ be the optimal test channel achieving $\rvR(R_1,D_1,D_2|P_X)$ in~\eqref{minimumr2} (assuming it is unique)\footnote{If optimal test channels are not unique, then following the proof of \cite[Lemma 2]{watanabe2015second}, we can argue that the tilted information density is still well defined.}. Let $P_{XY}^*,P_{XZ}^*,P_{YZ}^*$, $P_{Y}^*$ and $P_{Y|X}^*$ be the induced (conditional) marginal distributions. We are now ready to define the tilted information density for successive refinement source coding problem.

\begin{definition}
Given a boundary rate pair $(R_1^*,R_2^*)$ and distortion pair $(D_1,D_2)$, define the tilted information density as
\begin{align}
\label{srtilted}
&\nn\jmath_{YZ}(x,R_1^*,D_1,D_2|P_X)\\*
&:=-\log\mathbb{E}_{P_{YZ}^*}\left[\exp\left(-\lambda^*\left(\log\frac{P_{Y|X}^*(Y|x)}{P_Y^*(Y)}-R_1^*\right)-\nu_1^*(d_1(x,Y)-D_1)-\nu_2^*(d_2(x,Z)-D_2)\right)\right].
\end{align}
\end{definition}
We remark that for a successively refinable discrete memoryless source-distortion measure triplet, $\lambda^*=0$, $\nu_1^*=0$. Thus~\eqref{srtilted} reduces to the usual distortion-tilted information density~\eqref{kostinatilt}.

The properties of $\jmath_{YZ}(x,R_1^*,D_1,D_2|P_X)$ are summarized in the following lemma.
\begin{lemma}
\label{propertytilted}
The tilted information density $\jmath_{YZ}(x,R_1^*,D_1,D_2|P_X)$ has the following properties:
\begin{align}
\rvR(R_1^*,D_1,D_2|P_X)=\mathbb{E}_{P_X}\left[\jmath_{YZ}(X,R_1^*,D_1,D_2|P_X)\right]\label{expectlemma},
\end{align}
and for $P_{YZ}^*$-almost every $(y,z)$ and $\lambda^*>0$,
\begin{align}
 \jmath_{YZ}(x,R_1^*,D_1,D_2|P_X) &=\log\frac{P_{YZ|X}^*(y,z|x)}{P_{YZ}^*(y,z)} +\lambda^*\left(\log\frac{P_{Y|X}^*(y|x)}{P_Y^*(y)} -R_1^*\right)\nn\\* & \qquad-\nu_1^*(d_1(x,y)-D_1)-\nu_2^*(d_2(x,z)-D_2)\label{expand}. 
\end{align}
\end{lemma}
The proof of Lemma~\ref{propertytilted} is similar to~\cite[Lemma 1]{watanabe2015second} and given in Appendix~\ref{proofpropertytilted}. We remark that for a successively refinable discrete memoryless source-distortion measure triplet,~\eqref{expand} is replaced by~\cite[Property 1]{kostina2012converse}.

We can also relate $\jmath_{YZ}(x,R_1^*,D_1,D_2|P_X)$ to the derivative of $\rvR(R_1^*,D_1,D_2|Q_X)$ with respect to the source distribution $Q_X$ for some $Q_X$ in the neighborhood of $P_X$.
\begin{lemma}
\label{derivativer2}
Suppose that for all $Q_X$ in the neighborhood of $P_X$, $\mathrm{supp}(Q_{YZ}^*)=\mathrm{supp}(P_{YZ}^*)$. Then for all $a\in\calX$,
\begin{align}
\frac{\partial \rvR(R_1^*,D_1,D_2|Q_X)}{\partial Q_X(a)}\bigg|_{Q_X=P_X}=\jmath_{YZ}(a,R_1^*,D_1,D_2|P_X)-(1+\lambda^*).
\end{align}
\end{lemma}
The proof of Lemma~\ref{derivativer2} is similar to~\cite[Theorem 2.2]{kostina2013lossy} and given in Appendix~\ref{proofderivativer2}. For successively refinable discrete memoryless source-distortion measure triplets, the proof is exactly the same as~\cite[Theorem 2.2]{kostina2013lossy}. Hence, we remark that Lemma~\ref{derivativer2} is actually an extension of~\cite[Theorem 2.2]{kostina2013lossy}.

\subsection{General Discrete Memoryless Sources}
Define bivariate generalization of the Gaussian cdf as follows:
\begin{align}
\Psi(x,y,\bmu,\mathbf{\Sigma})&:=\int_{-\infty}^{x}\int_{-\infty}^{y}\calN(\bx; \bmu;\bSigma)\, \rmd \bx.
\end{align}
Here, $\calN(\bx; \bmu;\bSigma)$ is the pdf of a bivariate Gaussian with mean $\bmu$ and covariance matrix $\bSigma$~\cite[Chapter 1]{tan2015asymptotic}. Note that $\calN(\cdot ; \bmu;\bSigma)$  is a degenerate Gaussian if $\bSigma$ is singular.
For example if $\mathrm{rank}(\bSigma) =1 $, all the probability mass of the distribution $\calN(\cdot ; \bmu;\bSigma)$  lies on an affine subspace of dimension $1$ in $\bbR^2$.  As such, $\Psi(x,y,\bmu,\mathbf{\Sigma})$ is well-defined even if $\bSigma$ is singular.

Let $\mathrm{V}(D_1|P_X):=\mathrm{Var}[\jmath_Y(X,D_1|P_X)]$ and $\mathrm{V}(D_2|P_X):=\mathrm{Var}[\jmath_Z(X,D_2|P_X)]$ be {\em rate-dispersion functions}~\cite{kostina2012fixed}. 
Given a   rate pair $(R_1^*,R_2^*)$ on the boundary of $\calR(D_1,D_2|P_{X})$, also define another  rate-dispersion function $\mathrm{V}(R_1^*,D_1,D_2|P_X):=\mathrm{Var}\left[\jmath_{YZ}(X,R_1^*,D_1,D_2|P_X)\right]$. 
Let $\mathbf{V}(R_1^*,D_1,D_2|P_X) \succeq 0$ be the covariance matrix of the  two-dimensional random vector $[\jmath_Y(X,D_1|P_X),\jmath_{YZ}(X,R_1^*,D_1,D_2|P_X)]^T$, i.e., the {\em rate-dispersion matrix}.

We impose the following conditions on the rate pair $(R_1^*,R_2^*)$, the distortion measures $(d_1,d_2)$, the distortion levels $(D_1,D_2)$ and the source distribution $P_X$:
\begin{enumerate}
\item \label{cond1} $\rvR(R_1^*,D_1,D_2|P_X)$ is finite;
\item $\lambda^*\geq 0$ in~\eqref{deflambda} and $\nu_i^*,~i=1,2$ in \eqref{defnu1}, \eqref{defnu2} are well-defined (i.e., the derivatives exist);
\item $(Q_X,D_1')\mapsto R_Y(Q_X,D_1')$ is twice differentiable in the neighborhood of $(P_X,D_1)$ and the derivative is bounded (i.e., the spectral norm of the Hessian matrix is bounded);
\item \label{cond3} $(R_1,D_1',D_2',Q_{X})\mapsto \rvR(R_1,D_1',D_2'|Q_{X})$ is twice differentiable in the neighborhood of $(R_1^*,D_1,D_2,P_{X})$ and the derivative is bounded;
\end{enumerate}
We note that similar regularity assumptions were made in other works on second-order asymptotics for lossy source coding~\cite{ingber2011} and lossy joint source-channel coding~\cite{wang2011}. 

\begin{theorem}
\label{mainresult}
Under conditions (\ref{cond1}) to (\ref{cond3}), depending on the values of $(R_1^*,R_2^*)$, the optimal second-order $(R_1^*,R_2^*,D_1,D_2,\epsilon)$ coding region is as follows:
\begin{itemize}
\item Case (i): $R_Y(P_X,D_1)<R_1^*<\rvR(R_1^*,D_1,D_2|P_X)$ and $R_2^*=\rvR(R_1^*,D_1,D_2|P_X)$
\begin{align}
\calL(R_1^*,R_2^*,D_1,D_2,\epsilon)=\left\{(L_1,L_2):\lambda^*L_1+L_2\geq \sqrt{\mathrm{V}(R_1^*,D_1,D_2|P_X)}\mathrm{Q}^{-1}(\epsilon)\right\}.
\end{align}
\item Case (ii): $R_1^*=R_Y(P_X,D_1)$ and $R_2^*>\rvR(R_1^*,D_1,D_2|P_X)$
\begin{align}
\calL(R_1^*,R_2^*,D_1,D_2,\epsilon)=\left\{(L_1,L_2):L_1 \geq \sqrt{\mathrm{V}(D_1|P_X)} \rm\rmQ^{-1}(\epsilon)\right\}.
\end{align}
\item Case (iii): $R_1^*=R_Y(P_X,D_1)$, $R_2^*=\rvR(R_1^*,D_1,D_2|P_X)$  and $\mathrm{rank}(\mathbf{V}(R_1^*,D_1,D_2|P_X))\ge 1$,
\begin{align}
\calL(R_1^*,R_2^*,D_1,D_2,\epsilon)=\left\{(L_1,L_2): \Psi(L_1,\lambda^*L_1+L_2,\mathbf{0},\mathbf{V}(R_1^*,D_1,D_2|P_X))\geq 1-\epsilon\right\}. \label{eqn:calLiii}
\end{align}
\end{itemize} 
\end{theorem}

The proof of Theorem~\ref{mainresult} is provided  in Section~\ref{secondproof}.  A few remarks are in order. 

First, in both Cases (i) and (ii), the code is operating at a rate bounded away from one of the first-order fundamental limits. Hence, a univariate Gaussian suffices to characterize the second-order behavior. In contrast, for Case (iii), the code is operating at precisely the two first-order fundamental limits. Hence,  in general, we need a bivariate Gaussian to characterize the second-order behavior. Using an argument by Tan and Kosut~\cite[Theorem 6]{tan2014dispersions}, we note that  this result holds for both positive definite and rank deficient  rate-dispersion matrices $\mathbf{V}(R_1^*,D_1,D_2|P_X)$. However, we exclude the  degenerate case in which $\mathrm{rank}(\mathbf{V}(R_1^*,D_1,D_2|P_X))=0$. Note that if the rank of $\bV(R_1^*,D_1,D_2|P_X)$ is $0$, it means that the dispersion matrix is all zeros matrix, i.e., $\rmV(D_1|P_X)=0$, $\rmV(R_1^*,D_1,D_2|P_X)=0$ and $\mathrm{Cov}[\jmath_Y(x,D_1|P_X),\jmath_{YZ}(x,R_1^*,D_1,D_2|P_X)]=0$. This implies that $\jmath_Y(x,D_1|P_X)$ and $\jmath_{YZ}(x,R_1^*,D_1,D_2|P_X)$ are both deterministic random variables. In this case, the second-order term (dispersion) vanishes and if one seeks refined asymptotic estimates for the optimal finite blocklength  coding rates, one  would then be interested to analyze the {\em third-order} or $\Theta(\log n)$ asymptotics (cf.\ \cite[Theorem 18]{kostina2012fixed}). This, however, is beyond the scope of the present work.

Second, in Section \ref{srsource}, we illustrate the region in \eqref{eqn:calLiii} for successively refinable source-distortion measure triplets where the computation of $\mathbf{V}(R_1^*,D_1,D_2|P_X)$ is simplified. In principle, we can numerically evaluate the region $\calL(R_1^*,R_2^*,D_1,D_2,\epsilon)$ for non-successively refinable source-distortion measure triplets such as the one identified by Equitz and Cover in~\cite[Section IV]{equitz1991successive}, which is based on Gerrish's problem~\cite{gerrish}. However, the computations of $\rvR(R_1,D_1,D_2|P_X)$ (defined in \eqref{minr2}) and the optimal test channel $P_{YZ|X}^*$ are numerically unstable using off-the-shelf convex optimization software such as CVX~\cite{cvx}. One may need to develop specialized  Blahut-Arimoto-type algorithms~\cite[Chapter~8]{csiszar2011information} to solve for the optimal test channel. This is again  beyond the scope of this paper.

We are now ready to present our moderate deviation result. Define
\begin{align}
\theta=\lambda^*\theta_1+\theta_2.
\end{align}
\begin{theorem}
\label{mdconstant}
Given a rate pair $(R_1^*,R_2^*)\in\calR(D_1,D_2|P_X)$ satisfying that the conditions in Theorem~\ref{mainresult}, under the assumptions that $\mathrm{V}(D_1|P_X)>0$ and $\mathrm{V}(R_1^*,D_1,D_2|P_X)>0$, depending on the values of $(R_1^*,R_2^*)$, we have
\begin{itemize}
\item Case (i): $R_Y(P_X,D_1)<R_1^*<\rvR(R_1^*,D_1,D_2|P_X)$ and $R_2^*=\rvR(R_1^*,D_1,D_2|P_X)$
\begin{align}
\nu^*(R_1^*,R_2^*|D_1,D_2)=\frac{\theta^2}{2\mathrm{V}(R_1^*,D_1,D_2|P_X)}.
\end{align}
\item Case (ii): $R_1^*=R_Y(P_X,D_1)$ and $R_2^*>\rvR(R_1^*,D_1,D_2|P_X)$
\begin{align}
\nu^*(R_1^*,R_2^*|D_1,D_2)=\frac{\theta_1^2}{2\mathrm{V}(D_1|P_X)}.
\end{align}
\item Case (iii): $R_1^*=R_Y(P_X,D_1)$ and $R_2^*=\rvR(R_1^*,D_1,D_2|P_X)$
\begin{align}
\nu^*(R_1^*,R_2^*|D_1,D_2)=\min\left\{\frac{\theta_1^2}{2\mathrm{V}(D_1|P_X)},\frac{\theta^2}{2\mathrm{V}(R_1^*,D_1,D_2|P_X)}\right\}.
\end{align}
\end{itemize} 
\end{theorem}
Again a few remarks are in order. 

First, Theorem~\ref{mdconstant}  can be proved similarly as in~\cite{tan2012moderate} using Euclidean Information Theory~\cite{borade2008}. However, in Section \ref{proofmdc}, we use the moderate deviations principle/theorem (cf.\ Dembo and Zeitouni~\cite[Theorem~3.7.1]{dembo2009large}). We remark that the moderate deviations result for DMSes in Tan~\cite{tan2012moderate} for the point-to-point lossy source coding problem requires that $\frac{n\rho_n^2}{\log n}\to\infty$ as $n\to\infty$. However, our proof only requires the  condition that $n\rho_n^2\to \infty$ as $n\to\infty$. 
The additional $\log n$ in the condition for the sequence $\{\rho_n\}_{n\ge 1}$ in \cite{tan2012moderate} results from the fact that the proof therein is based heavily on the method of types and the type counting lemma. Instead, if we use the information spectrum method together with  properties of the $D$-tilted information density (cf.\ Kostina and Verd\'u~\cite{kostina2012fixed}), we only require that $n\rho_n^2\to\infty$. Furthermore, Tan's result in \cite{tan2012moderate} is  a corollary of Theorem  \ref{mdconstant} since the point-to-point lossy source coding problem is a special case of the successive refinement problem.

Second, from both theorems, we observe that the rate-dispersion functions  $\mathrm{V}(R_1^*,D_1,D_2|P_X),\mathrm{V}(D_1|P_X)$ and the rate-dispersion matrix $\mathbf{V}(R_1^*,D_1,D_2|P_X)$ are essential in characterizing the fundamental limits of the successive refinement problem. 

Third, we remark that similar results (at least for  the achievability part) can be established under the separate excess-distortion probabilities criterion~\cite{no2015strong}. We discuss this in greater detail after Corollary \ref{srmdc} in the sequel for successively refinable discrete memoryless sources for which the converse is implied by the point-to-point lossy source coding results.

Finally, we remark that the two rate-dispersion functions $\rmV(D_1|P_X)$ and $\rmV(R_1^*,D_1,D_2|P_X)$ can be related with the error exponent functions in \cite{kanlis1996error}, similarly to  how the channel dispersion and the channel coding error exponent are connected \cite{altugwagner2014}. In particular, in \cite[Proposition 2]{ingber2011} (see also \cite[Lemma 2]{tan2012moderate}), it has been shown that
\begin{align}
\rmV(D_1|P_X)=\left[\frac{\partial^2 F(R_1,D_1|P_X)}{\partial R_1^2}\bigg|_{R_1=R_Y(P_X,D_1)}\right]^{-1},
\end{align}
where $F(R_1,D_1|P_X)$ is Marton's exponent for lossy source coding~\cite{Marton74}. In a completely analogous manner, one can show that
\begin{align}
\rmV(R_1^*,D_1,D_2|P_X)=\left[\frac{\partial^2 F_{\rmc}(R_2,D_1,D_2|R_1^*,P_X)}{\partial R_2^2}\bigg|_{R_2=\rvR(R_1^*,D_1,D_2|P_X)}\right]^{-1},
\end{align}
where $F_{\rmc}(R_2,D_1,D_2|R_1^*,P_X)$ is the  conditional  error exponent for the second decoder in successive refinement problem~\cite{kanlis1996error}. Note that in \cite{kanlis1996error}, the error exponent under the joint excess-distortion probability criterion is given by the minimum of $F(R_1,D_1|P_X)$ and $F_{\rmc}(R_2,D_1,D_2|R_1^*,P_X)$. Hence, from case (iii) in Theorem \ref{mdconstant}, we conclude that our moderate deviations constant result is parallel to  the error exponent result in \cite{kanlis1996error}.

\subsection{Successively Refinable Discrete Memoryless Sources}
\label{srsource}
In this subsection, we specialize the results in Theorem~\ref{mainresult} and~\ref{mdconstant} to successively refinable discrete memoryless source-distortion measure triplets. Note that for such source-distortion measure triplets, $\rvR(R_1^*,D_1,D_2|P_X)=R_Z(P_X,D_2)$ if $R_Y(P_X,D_1)\leq R_1^*<R_Z(P_X,D_2)$. Hence, $\lambda^*=0$ and $\nu^*_1=0$ and $\jmath(X,R_1^*,D_1,D_2|P_X)=\jmath_Z(X,D_2|P_X)$. The covariance matrix $\mathbf{V}(R_1^*,D_1,D_2|P_X)$ is also simplified to $\mathbf{V}(D_1,D_2|P_X)$ with diagonal elements being $\mathrm{V}(D_1|P_X)$ and $\mathrm{V}(D_2|P_X)$ and off-diagonal element being $\mathrm{Cov}[\jmath_Y(X,D_1|P_X),\jmath_Z(X,D_2|P_X)]$. The conditions in Theorem~\ref{mainresult} are also now simplified to: $(Q_X,D_1')\mapsto R_Y(Q_X,D_1')$ and $(Q_X,D_2')\mapsto R_Z(Q_X,D_2')$  are twice differentiable in the neighborhood of $(P_X,D_1,D_2)$ and the derivatives are bounded.
\begin{corollary}
\label{srmainresult}
Under the conditions stated above, depending on $(R_1^*,R_2^*)$, the optimal second-order $(R_1^*,R_2^*,D_1,D_2,\epsilon)$ coding region for a successively refinable discrete memoryless source-distortion measure triplet is as follows:
\begin{itemize}
\item Case (i): $R_Y(P_X,D_1)<R_1^*<R_Z(P_X,D_2)$ and $R_2^*=R_Z(P_X,D_2)$
\begin{align}
\calL(R_1^*,R_2^*,D_1,D_2,\epsilon)=\left\{(L_1,L_2):L_2 \geq \sqrt{\mathrm{V}(D_2|P_X)} \rm\rmQ^{-1}(\epsilon)\right\}.
\end{align}
\item Case (ii): $R_1^*=R_Z(P_X,D_2)$ and $R_2^*>R_Z(P_X,D_2)$
\begin{align}
\calL(R_1^*,R_2^*,D_1,D_2,\epsilon)=\left\{(L_1,L_2):L_1 \geq \sqrt{\mathrm{V}(D_1|P_X)} \rm\rmQ^{-1}(\epsilon)\right\}.
\end{align}
\item Case (iii): $R_1^*=R_Z(P_X,D_2)$ and $R_2^*=R_Z(P_X,D_2)$ and $\mathrm{rank}(\mathbf{V}(D_1,D_2|P_X))\geq 1$,
\begin{align}
\calL(R_1^*,R_2^*,D_1,D_2,\epsilon)=\left\{(L_1,L_2): \Psi(L_1,L_2,\mathbf{0},\mathbf{V}(D_1,D_2|P_X))\geq 1-\epsilon\right\}.\label{eqn:pd}
\end{align}
Specifically, if $\mathbf{V}(D_1,D_2|P_X)=\mathrm{V}(D_1|P_X)\cdot \mathrm{ones}(2,2)$, or equivalently $\jmath_{Y}(X,D_1|P_X)-R_1^*=\jmath_Z(X,D_2|P_X)-R_2^*$ almost surely
\begin{align}
\calL(R_1^*,R_2^*,D_1,D_2,\epsilon)=\left\{(L_1,L_2):\min\{L_1,L_2\} \geq \sqrt{\mathrm{V}(D_1|P_X)} \rm\rmQ^{-1}(\epsilon)\right\}\label{eqn:ones}.
\end{align}
\end{itemize} 
\end{corollary}
Corollary~\ref{srmainresult} results from specializations of Theorem~\ref{mainresult}. The special case in \eqref{eqn:ones} is proved in Section~\ref{proofsrmain}. We notice that the expressions in the second-order regions are simplified for successively refinable discrete memoryless source-distortion measure triplets. In particular, the optimization to compute the optimal test channel $P_{YZ|X}^*$ in $\rvR(R_1, D_1, D_2|P_X)$, defined in \eqref{minimumr2}--\eqref{minr2}, is no longer necessary since the Markov chain $X-Z-Y$ holds for $P_{YZ|X}^*$~\cite{equitz1991successive}.

The case in~\eqref{eqn:ones} pertains, for example, to a binary source with Hamming distortion measures. For such a source-distortion measure triplet, $\mathbf{V}(D_1,D_2|P_X)$ is rank $1$ and proportional to the all ones matrix. See Subsection \ref{sec:bin}. The result in~\eqref{eqn:ones} implies that both excess-distortion events in~\eqref{defexcessprob} are perfectly correlated so the one consisting of the {\em smaller} second-order rate $L_i,~i=1,2$ dominates, since the first-order rates are fixed at the first-order fundamental limits $(R_Y(P_X, D_1),R_Z(P_X, D_2))$. In fact, our result in \eqref{eqn:ones} specializes to the scenario where one considers the {\em separate} excess-distortion criterion~\cite{no2015strong} in~\eqref{eqn:eta1}--\eqref{eqn:eta2} with $\eta_1=\eta_2=\epsilon$ and $\rmV(D_1|P_X)=\rmV(D_2|P_X)$.   More importantly, the case in \eqref{eqn:pd} when $\mathbf{V}(D_1,D_2|P_X)$ is full rank pertains to a source-distortion measure triplets with more ``degrees-of-freedom''. See Subsection \ref{sec:quat} for a concrete example.  Thus our work is a strict generalization of that in \cite{no2015strong}.

\begin{corollary}
\label{srmdc}
Under the conditions in Theorem~\ref{srmainresult} and the assumptions that $\mathrm{V}(D_i|P_X)>0,~i=1,2$, depending on $(R_1^*,R_2^*)$, the moderate deviations constant for a successively refinable discrete memoryless source-distortion measure triplet is as follows:
\begin{itemize}
\item Case (i): $R_Y(P_X,D_1)<R_1^*<R_Z(P_X,D_2)$ and $R_2^*=R_Z(P_X,D_2)$
\begin{align}
\nu^*(R_1^*,R_2^*|D_1,D_2)=\frac{\theta_2^2}{2\rmV(D_2|P_X)}.
\end{align}
\item Case (ii): $R_1^*=R_Z(P_X,D_2)$ and $R_2^*>R_Z(P_X,D_2)$
\begin{align}
\nu^*(R_1^*,R_2^*|D_1,D_2)=\frac{\theta_1^2}{2\rmV(D_1|P_X)}.
\end{align}
\item Case (iii): $R_1^*=R_Z(P_X,D_2)$ and $R_2^*=R_Z(P_X,D_2)$
\begin{align}
\nu^*(R_1^*,R_2^*|D_1,D_2)=\min\left\{\frac{\theta_1^2}{2\rmV(D_1|P_X)},\frac{\theta_2^2}{2\rmV(D_2|P_X)}\right\}.
\end{align}
\end{itemize} 
\end{corollary}
Corollary~\ref{srmdc} is a simple specialization of Theorem~\ref{mdconstant}.

We remark that if we consider the separate excess-distortion probability criterion, similar moderate deviations results  can be established. Recall that  the separate excess-distortion probabilities are defined as in \eqref{eqn:eta1} and \eqref{eqn:eta2}.
An $(R_1^*,R_2^*)$-achievable moderate deviations constant pair $(\nu_1,\nu_2)$ can be defined similarly as Definition \ref{defmdconstant} except that we replace \eqref{jmdcerror} with the following two constraints:
\begin{align}
\liminf_{n\to\infty}-\frac{\log \epsilon_{i,n}(D_i)}{n\rho_n^2}\geq \nu_i,~i=1,2.
\end{align}
We denote the  closure of all $(R_1^*,R_2^*)$-achievable moderate deviations constants pairs as $\calV_{\mathrm{sep}}(R_1^*,R_2^*|D_1,D_2)$. Following similar proof techniques as Theorem \ref{mdconstant}, one can easily conclude that
\begin{itemize}
\item Case (i):
\begin{align}
\calV_{\mathrm{sep}}(R_1^*,R_2^*|D_1,D_2):=\Big\{(\nu_1,\nu_2)
\in\bbR_+^2:  \nu_2\leq \frac{\theta_2^2}{2\rmV(D_2|P_X)}\Big\}
\end{align}
\item Case (ii):
\begin{align}
\calV_{\mathrm{sep}}(R_1^*,R_2^*|D_1,D_2):=\Big\{(\nu_1,\nu_2)
\in\bbR_+^2:\nu_1\leq \frac{\theta_1^2}{2\rmV(D_1|P_X)}\Big\}
\end{align}
\item Case (iii):
\begin{align}
\calV_{\mathrm{sep}}(R_1^*,R_2^*|D_1,D_2):=\Big\{(\nu_1,\nu_2)
\in\bbR_+^2:\nu_1\leq \frac{\theta_1^2}{2\rmV(D_1|P_X)},\nu_2\le\frac{\theta_2^2}{2\rmV(D_2|P_X)}\Big\}.
\end{align}
\end{itemize}
The above result is tight since the converse is implied by the converse  for the point-to-point lossy source coding problem~\cite{ingber2011}.

\subsection{Numerical Examples}
\label{sec:examples}
Recall that any discrete memoryless source with Hamming distortion measures is successively refinable~\cite{equitz1991successive}. In this subsection, we present two numerical examples from Kostina and Verd\'u~\cite{kostina2012fixed} to illustrate Corollary~\ref{srmainresult}. To be consistent with~\cite{kostina2012fixed}, we will use logarithm with base $2$ in this subsection.
\subsubsection{A Binary Memoryless Source with Hamming Distortion Measures} \label{sec:bin}
Fix $p\in[0,1]$. We consider a binary source with $P_X(0)=p$. For any distortion levels $D_2<D_1<p$, we obtain from~\cite[Example 1]{kostina2012fixed} that
\begin{align}
 \jmath_Y(x,D_1|P_X)=\imath_{P_X}(x)- h(D_1),\\
 \jmath_Z(x,D_2|P_X)=\imath_{P_X}(x)- h(D_2),
\end{align}
where $\imath_{P_X}(x)=\log\frac{1}{P_X(x)}$ and
$h(x):=-x\log x-(1-x)\log(1-x)$ is the binary entropy function. Hence,
\begin{align}
\mathrm{V}(D_1|P_X)=\mathrm{V}(D_2|P_X)=\mathrm{Cov}[\jmath_Y(X,D_1|P_X),\jmath_Z(X,D_2|P_X)].
\end{align}
and the rate-dispersion matrix is
\begin{align}
\mathbf{V}(D_1,D_2|P_X)=\mathrm{V}(D_1|P_X)\cdot \mathrm{ones}(2,2)
 = p(1-p)\log^2\left(\frac{1-p}{p}\right) \cdot \mathrm{ones}(2,2) , \label{eqn:binary}
\end{align} which does not depend on $(D_1, D_2)$. 
From the above considerations, we see that a binary source with Hamming distortion measures is an example that falls under \eqref{eqn:ones} in Corollary~\ref{srmainresult}.
\subsubsection{A Quaternary Memoryless Source with Hamming Distortion Measures} \label{sec:quat}
\begin{figure}[t]
\centering 
\includegraphics[width=10cm]{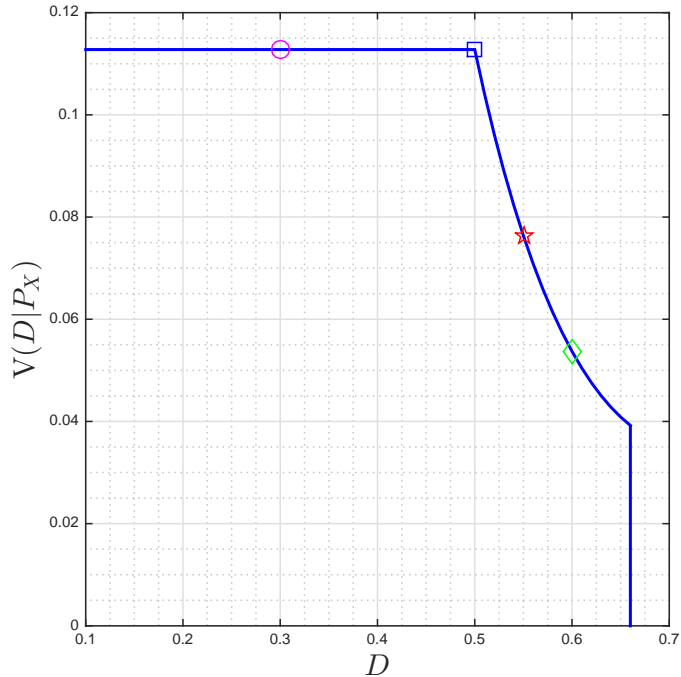}
\caption{Rate-dispersion  function $\rmV(D|P_X)$  for the source $P_X=[1/3,1/4,1/4,1/6]$~\cite[Section~VII.B]{kostina2012fixed} as a function of the distortion $D$.}
\label{rate_dispersion}
\end{figure}

We now consider a more interesting source with joint excess-distortion probability upper bounded by $\epsilon=0.005$. In particular, we consider a quaternary memoryless source with distribution $P_X=[1/3,1/4,1/4,1/6]$. This example illustrates Case (iii)  of Corollary \ref{srmainresult} and is adopted from~\cite[Section~VII.B]{kostina2012fixed}. The expressions for  the rate-distortion function  and the distortion-tilted information density are given in~\cite[Section~VII.B]{kostina2012fixed} (and will not be reproduced here as they are not important for our discussion). Since $\jmath_Y(x,D_1|P_X)=\jmath_Z(x,D_2|P_X)$ when $D_1=D_2=D$, we use $\jmath(x,D|P_X)$ to denote the common value of the distortion-tilted information density. Similarly, let $\mathrm{V}(D|P_X)$ be the common value of $\mathrm{V}(D_1|P_X)$ and $\mathrm{V}(D_2|P_X)$ when $D_1=D_2=D$. As shown in Figure~\ref{rate_dispersion} (reproduced from~\cite[Section~VII.B,~Figure~4]{kostina2012fixed}), the rate-dispersion function $\mathrm{V}(D|P_X)$ is dependent on  the distortion level $D$, unlike the binary example in Section \ref{sec:bin}.
\begin{figure}[t]
\centering
\includegraphics[width=10cm]{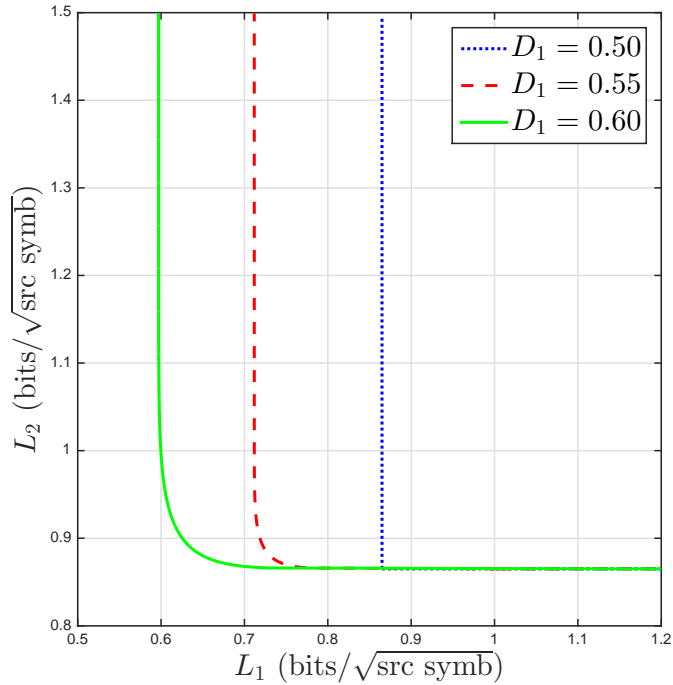}
\caption{Boundaries of the second-order coding region $\calL(R_1^*,R_2^*,D_1,D_2,\epsilon)$ for Case (iii) in Corollary~\ref{srmainresult}. The regions are to the top right of the boundaries.}
\label{plot_srregion}
\end{figure}

In this numerical example, we fix $D_2=0.3$, which is denoted by the circle in Figure~\ref{rate_dispersion}. Then we decrease $D_1$ from $0.6$ to $0.55$ and finally to $0.5$. These points are denoted respectively by the diamond, the pentagram and the square in Figure~\ref{rate_dispersion}. Given these values of $(D_1,D_2)$, we plot the second-order coding rate for Case (iii) of Corollary~\ref{srmainresult} in Figure~\ref{plot_srregion}. 

From Figure~\ref{plot_srregion}, we make the following observations and conclusions.
\begin{itemize}
\item The minimum $L_1$ converges to $\sqrt{V(D_1|P_X)}\rmQ^{-1}(\epsilon)$ as $L_2\uparrow \infty$. This is because large $L_2$, the bivariate Gaussian cdf degenerates to the univariate Gaussian cdf with mean $0$ and variance $\mathrm{V}(D_1|P_X)$. 

\item As we decrease the value of $D_1$, the second-order coding region shrinks. We remark that there is  a  transition from \eqref{eqn:pd} with $\mathrm{rank}(\bV(D_1,D_2|P_X))=2$ to \eqref{eqn:ones} (where $\mathrm{rank}(\bV(D_1,D_2|P_X))=1$) as we decrease $D_1$ with the critical value of $D_1$ being $0.5$.

\item When $D_2<D_1\le 0.5$, the rate-dispersion matrix $\mathbf{V}(D_1,D_2|P_X)$ is rank $1$ (and proportional to the all ones matrix). Correspondingly, the result in \eqref{eqn:ones} applies. Here,  the second-order region is a (unbounded) rectangle with a  sharp corner at the left bottom since the smaller $L_i,~i=1,2$ dominates. The second-order region remains unchanged as we decrease $D_1$ towards $D_2$ for fixed $D_2=0.3$. 

\item When $0.5<D_1<2/3$, the result in \eqref{eqn:pd}  with $\mathrm{rank}(\bV(D_1,D_2|P_X))=2$ applies. In this case, neither $L_1$ nor $L_2$ dominates. The second-order coding rates $(L_1,L_2)$  are coupled  together by the  full rank   rate-dispersion matrix $\mathbf{V}(D_1,D_2|P_X)$, resulting the smooth boundary at the left bottom. 
 
\end{itemize}

We conclude that depending on the value of the distortion levels, the rate-dispersion matrix is rank $1$ or rank $2$, illustrating Case (iii) of Corollary~\ref{srmainresult}. These interesting observations cannot be gleaned from the work of No, Ingber and Weissman~\cite{no2015strong} in which the separate excess-distortion criteria are employed for the successive refinement problem. When $\bV(D_1,D_2|P_X)$ is rank $1$, exactly one excess-distortion event dominates the probability in~\eqref{defexcessprob} entirely; when $\bV(D_1,D_2|P_X)$ is rank $2$, both excess-distortion events contribute non-trivially to the probability and a bivariate Gaussian is required to characterize the second-order fundamental limit.

\subsection{A One-Shot Converse Bound and an Alternative Converse Proof of Corollary~\ref{srmainresult}}
To conclude this section, we present a one-shot converse bound which generalizes the one-shot lower bounds on the excess-distortion probabilities for point-to-point lossy source coding and source coding with side information in~\cite{kostina2012converse}. Note that this converse bound is not useful to prove to the converse part for the general DMS case (of non-successively refinable source-distortion measure triplets) in Theorem \ref{mainresult}. For that we need to use a  strong converse technique of Gu and Effros~\cite{wei2009strong}, leading to the type-based ``strong converse'' in Lemma \ref{typestrongconverse}. However, this one-shot converse may be of independent interest (to other multi-terminal rate-distortion problems) and leads immediately to the converse parts of Corollary~\ref{srmainresult}.

\begin{lemma}
\label{tiltedconverse}
For any $(n,M_1,M_2)$-code  for the successive refinement problem with $n=1$ and any $\gamma_1\geq 0,~\gamma_2\geq 0$, we have
\begin{align}
\epsilon_n(D_1,D_2)
\geq \Pr\big(\jmath_Y(X,D_1|P_X) &\geq \log M_1+\gamma_1~\mathrm{or}\nn\\*
~\jmath_Z(X,D_2|P_X)&\geq \log(M_1M_2)+\gamma_2\big)\nn\\*
&\hspace{-.5in}-\exp(-\gamma_1)-\exp(-\gamma_2).
\end{align}
\end{lemma}
The proof of Lemma is provided in Appendix~\ref{prooftiltedconverse}.

For a memoryless source $X^n$, it is clear that
\begin{align}
\jmath_Y(X^n,D_1|P_{X}^n)
&=\sum_{i=1}^n\jmath_Y(X_i,D_1|P_X),\label{eqn:tiltex1}
\end{align}
and similarly for $\jmath_Z(X^n,D_2|P_{X}^n)$.  
Let $\gamma_1=\gamma_2=\frac{1}{2}\log n$. Let $\log M_1=nR_1^*+L_1\sqrt{n}-\gamma_1$ and $\log(M_1M_2)=nR_2^*+L_2\sqrt{n}-\gamma_2$. Invoking Lemma~\ref{tiltedconverse}, we obtain
\begin{align}
1-\epsilon_n(D_1,D_2)
\leq \Pr\left(\sum_{i=1}^n\jmath_Y(X_i,D_1|P_X)<nR_1^*+L_1\sqrt{n},~\sum_{i=1}^n\jmath_Z(X_i,D_2|P_X)<nR_2^*+L_2\sqrt{n}\right)+\frac{2}{\sqrt{n}}\label{ntiltedconverse}.
\end{align}
The rest of the proof is similar to the converse proof of Corollary~\ref{srmainresult} (in Section~\ref{proofsrmain}). We remark that this alternative converse proof holds also sources with arbitrary alphabets such as a GMS with quadratic distortion measures and a Laplacian source with absolute distortion measures~\cite{zhong2006type}. Indeed, we use this one-shot converse bound to prove the converse part of our Gaussian results in Sections~\ref{sec:prf_g_2} and \ref{sec:prf_g_m}.

\section{Proof of Second-Order Asymptotics for A DMS}
\label{secondproof}
\subsection{Achievability Coding Theorem}
\label{secondach}
We make use of the type covering lemma~\cite[Lemma 8]{no2015strong}, which is modified from~\cite[Lemma 1]{kanlis1996error}. Leveraging  the type covering lemma, we can then upper bound the excess-distortion probability. Finally, we Taylor expand appropriate terms and invoke the Berry-Essen theorem to obtain an achievable second-order coding region.

Define two constants:
\begin{align}
c_1&=4|\calX|\cdot|\calY|+9,\\*
c_2&=6|\calX|\cdot|\calY|\cdot|\calZ|+2|\calX|\cdot|\calY|+17.
\end{align}
We are now ready to recall the \emph{discrete} type covering lemma for successive refinement source coding in~\cite{kanlis1996error} and~\cite{no2015strong}.
\begin{lemma}
\label{typecovering}
Given type $Q_X\in\calP_n(\calX)$, for all $R_1\geq R_Y(Q_X,D_1)$, the following holds:
\begin{itemize}
\item There exists a set $\calB_Y\subset\calY^n$ such that 
\begin{align}
\frac{1}{n}\log|\calB_{Y}|\leq R_1+c_1\frac{\log n}{n}
\end{align}
and $\calB_Y$ $D_1$-covers $\calT_{Q_X}$, i.e.,
\begin{align}
\calT_{Q_{X}}\subset\bigcup_{y^n\in\calB_{Y}}\calN_1(y^n,D_1),
\end{align}
where 
\begin{align}
\calN_{1}(y^n,D_1):=\left\{x^n:d_{1}(x^n,y^n)\leq D_1\right\}.
\end{align}
\item For each $x^n\in\calT_{Q}$ and each $y^n\in\calB_1$, there exists a set $\calB_{Z}(y^n)\subset\calZ^n$ such that
\begin{align}
\frac{1}{n}\log \left(\sum_{y^n\in\calB_{Y}}|\calB_{Z}(y^n)|\right)\leq \rvR(R_1,D_1,D_2|Q_X)+c_2\frac{\log n}{n}
\end{align}
and $\calB_{Z}(y^n)$ $D_2$-covers $\calN_1(y^n,D_1)$, i.e.,
\begin{align}
\calN_1(y^n,D_1)\subset\bigcup_{z^n\in\calB_{Z}(y^n)}\calN_2(z^n,D_2),
\end{align}
where
\begin{align}
\calN_2(z^n,D_2):=\left\{x^n:d_{2}(x^n,z^n)\leq D_2\right\}.
\end{align}
\end{itemize}
\end{lemma}

Invoking Lemma~\ref{typecovering}, we can then upper bound the excess-distortion probability for some $(n,M_1,M_2)$-code. Given any $(n,M_1,M_2)$-code, define
\begin{align}
R_{1,n}&:=\frac{1}{n}\Bigg(\log M_1-c_1\log n-|\calX|\log(n+1)\Bigg),\\
R_{2,n}&:=\frac{1}{n}\Bigg(\log(M_1M_2)-c_2\log n\Bigg).
\end{align}

\begin{lemma}
\label{uppexcess}
There exists an $(n,M_1,M_2)$-code such that
\begin{align}
\epsilon_n(D_1,D_2)\leq \Pr\left(R_{1,n}<R_Y(\hat{T}_{X^n},D_1)~\mathrm{or}~R_{2,n}<\rvR(R_{1,n},D_1,D_2|\hat{T}_{X^n})\right).
\end{align}
\end{lemma}
The proof of Lemma~\ref{uppexcess} is similar to~\cite[Lemma 5]{watanabe2015second} and given in Appendix~\ref{proofuppexcess}.

Define the typical set 
\begin{align}
\calA_n(P_{X}):=\left\{Q_{X}\in\calP_n(\calX):\left\|Q_X-P_X\right\|_{\infty}\leq \sqrt{\frac{\log n}{n}}\right\}.
\end{align}
According to~\cite[Lemma 22]{tan2014state},
\begin{align}
\Pr\left(\hat{T}_{X^n}\notin\calA_n(P_{X})\right)\leq \frac{2|\calX|}{n^2}.
\end{align}
For a rate pair $(R_1^*,R_2^*)$ satisfying the conditions in Theorem~\ref{mainresult}, we choose 
\begin{align}
\frac{1}{n}\log M_1&=R_1^*+\frac{L_1}{\sqrt{n}}+\frac{c_1\log n+|\calX|\log (n+1)}{n}\label{achm1},\\
\frac{1}{n}\log(M_1M_2)&=R_2^*+\frac{L_2}{\sqrt{n}}+c_2\frac{\log n}{n}\label{achm2}.
\end{align}

Hence,
\begin{align}
R_{i,n}=R_i^*+\frac{L_i}{\sqrt{n}}~,i=1,2.
\end{align}
From the conditions in Theorem~\ref{mainresult}, we know that the second derivative  of $R_Y(Q_X,D_1)$ is bounded in the neighborhood of $P_X$, and that the second derivative  of $\rvR(R_1,D_1,D_2|Q_X)$ with respect to $(R_1,R_2,Q_X)$ is bounded around a neighborhood of $(R_1^*,P_{X})$. Hence, for any $x^n$ such that $\hat{T}_{x^n}\in\calA_n(P_X)$, applying Taylor's expansion and invoking Lemma~\ref{derivativer2} and~\cite[Theorem 2.2]{kostina2013lossy}, we obtain
\begin{align}
R_Y(\hat{T}_{x^n},D_1)
&=R_Y(P_X,D_1)+\sum_{x}\left(\hat{T}_{x^n}(x)-P_{X}(x)\right)\jmath_X(x,D_1|P_X)+O\left(\frac{\log n}{n}\right),\\*
&=\frac{1}{n}\sum_{i=1}^n \jmath_Y(x_i,D_1|P_X)+O\left(\frac{\log n}{n}\right)\label{taylor_1},
\end{align}
and
\begin{align}
\nn&\rvR(R_{1,n},D_1,D_2|\hat{T}_{x^n})\\*
&=\rvR(R_1^*,D_1,D_2|P_{XY})-\lambda^*\frac{L_1}{\sqrt{n}}+\sum_{x}\left(\hat{T}_{x^n}(x)-P_{X}(x)\right)\jmath_{YZ}(x,R_1^*,D_1,D_2|P_{X})+O\left(\frac{\log n}{n}\right)\\*
&=\frac{1}{n}\sum_{i=1}^n\jmath_{YZ}(x_i,R_1^*,D_1,D_2|P_{X})-\lambda^*\frac{L_1}{\sqrt{n}}+O\left(\frac{\log n}{n}\right)\label{taylor_2}.
\end{align}
Define $\xi_n=\frac{\log n}{n}$. 

Hence, invoking Lemma~\ref{uppexcess}, for large $n$, we obtain
\begin{align}
\epsilon_n(D_1,D_2)
&\leq \Pr\left(R_{1,n}<R_Y(\hat{T}_{X^n},D_1)~\mathrm{or}~R_{2,n}<\rvR(R_{1,n},D_1,D_2|\hat{T}_{X^n})\right)\\
&\leq \Pr\left(R_{1,n}<R_Y(\hat{T}_{X^n},D_1)~\mathrm{or}~R_{2,n}< \rvR(R_{1,n},D_1,D_2|\hat{T}_{X^n}),\hat{T}_{X^n}\in\calA_{n}(P_X)\right)+\Pr\left(\hat{T}_{X^n}\notin\calA_{n}(P_X)\right)\\
&\nn\leq \Pr\left(R_1^*+\frac{L_1}{\sqrt{n}}<\frac{1}{n}\sum_{i=1}^n \jmath_Y(X_i,D_1|P_X)+O\left(\xi_n\right)~\mathrm{or}~\right.\\
&\qquad\left.R_2^*+\frac{L_2}{\sqrt{n}}<\frac{1}{n}\sum_{i=1}^n\jmath_{YZ}(X_i,R_1^*,D_1,D_2|P_{X})-\lambda^*\frac{L_1}{\sqrt{n}}+O(\xi_n)\right)+\frac{2|\calX|}{n^2}\\
&\nn=\Pr\left(R_1^*+\frac{L_1}{\sqrt{n}}<\frac{1}{n}\sum_{i=1}^n \jmath_Y(X_i,D_1|P_X)+O\left(\xi_n\right)~\mathrm{or}~\right.\\
&\qquad \left. R_2^*+\lambda^*\frac{L_1}{\sqrt{n}}+\frac{L_2}{\sqrt{n}}<\frac{1}{n}\sum_{i=1}^n\jmath_{YZ}(X_i,R_1^*,D_1,D_2|P_{X})+O(\xi_n)\right)+\frac{2|\calX|}{n^2}\label{taylorexpand}.
\end{align}

Therefore,
\begin{align}
\nn1-\epsilon_n(D_1,D_2)
&\geq 
\Pr\left(\frac{1}{n}\sum_{i=1}^n \jmath_Y(X_i,D_1|P_X)\leq R_1^*+\frac{L_1}{\sqrt{n}}+O\left(\xi_n\right),\right.\\*
&\qquad \left.\frac{1}{n}\sum_{i=1}^n\jmath_{YZ}(X_i,R_1^*,D_1,D_2|P_{X})\leq R_2^*+\lambda^*\frac{L_1}{\sqrt{n}}+\frac{L_2}{\sqrt{n}}+O(\xi_n)\right)-\frac{2|\calX|}{n^2}\label{eqn:dmsach}.
\end{align}
We consider Case (i) first where $R_Y(P_X,D_1)<R_1^*<\rvR(R_1^*,D_1,D_2|P_X)$ and $R_2^*=\rvR(R_1^*,D_1,D_2|P_X)$. Using the weak law of large numbers, we obtain
\begin{align}
\Pr\left(\frac{1}{n}\sum_{i=1}^n \jmath_Y(X_i,D_1|P_X)\leq R_1^*+\frac{L_1}{\sqrt{n}}+O\left(\xi_n\right)\right)\to 1.
\end{align}
Using the  Berry-Esseen Theorem, we obtain
\begin{align} \nn&\Pr\left(\frac{1}{n}\sum_{i=1}^n\jmath_{YZ}(X_i,R_1^*,D_1,D_2|P_{X})\leq R_2^*+\lambda^*\frac{L_1}{\sqrt{n}}+\frac{L_2}{\sqrt{n}}+O(\xi_n)\right)\\*
&\geq 1-\rmQ\left(\frac{\lambda^*L_1+L_2+O(\sqrt{n}\xi_n)}{\sqrt{\mathrm{V}(R_1^*,D_1,D_2|P_X)}}\right)-\frac{6\mathrm{T}(R_1^*,D_1,D_2|P_X)}{\sqrt{n}\mathrm{V}^{3/2}(R_1^*,D_1,D_2|P_X)}\label{berryesseen},
\end{align}
where  $\mathrm{T}(R_1^*,D_1,D_2|P_X)$ is the third absolute moment of $\jmath_{YZ}(X,R_1^*,D_1,D_2|P_{X})$, which is finite. 
Hence,
\begin{align}
\epsilon_n(D_1,D_2)\leq \rmQ\left(\frac{\lambda^*L_1+L_2+O(\sqrt{n}\xi_n)}{\sqrt{\mathrm{V}(R_1^*,D_1,D_2|P_X)}}\right)+\frac{6\mathrm{T}(R_1^*,D_1,D_2|P_X)}{\sqrt{n}\mathrm{V}^{3/2}(R_1^*,D_1,D_2|P_X)}+\frac{2|\calX|}{n^2}\label{phitoq}.
\end{align}

From the conditions in Theorem~\ref{mainresult}, we conclude that $\mathrm{T}(R_1^*,D_1,D_2|P_X)$ is finite. Hence, if $(L_1,L_2)$ satisfies 
\begin{align}
\lambda^*L_1+L_2\geq \sqrt{\mathrm{V}(R_1^*,D_1,D_2|P_X)}\rm\rmQ^{-1}(\epsilon),
\end{align}
then $\limsup_{n\to\infty}\epsilon_n(D_1,D_2)\leq \epsilon$. We omit the proof for Case (ii) since it is similar to Case (i). 

The most interesting case is Case (iii) where $R_1^*=R_Y(P_X,D_1)$ and $R_2^*=\rvR(R_1^*,D_1,D_2|P_X)$. If $\mathbf{V}(R_1^*,D_1,D_2|P_X)$ is positive definite we invoke the multi-variate Berry-Esseen Theorem~\cite{Ben03} to obtain
\begin{align}
\epsilon_n(D_1,D_2)\leq 1-\Psi\left(L_1+O\left(\xi_n\right),\lambda^*L_1+L_2+O\left(\xi_n\right),\mathbf{0},\mathbf{V}(R_1^*,D_1,D_2|P_X)\right)+O\left(\frac{1}{\sqrt{n}}\right). \label{eqn:multi-be}
\end{align}  
Note that if $\mathbf{V}(R_1^*,D_1,D_2|P_X)$ is rank $1$, we can use the argument (projection onto a lower-dimensional subspace) in \cite[Proof of Theorem 6]{tan2014dispersions} to conclude that \eqref{eqn:multi-be} also holds. 
Now if we choose $(L_1,L_2)$ such that
\begin{align}
\Psi\left(L_1,\lambda^*L_1+L_2,\mathbf{0},\mathbf{V}(R_1^*,D_1,D_2|P_X)\right)\geq 1-\epsilon,
\end{align}
then $\limsup_{n\to\infty}\epsilon_n(D_1,D_2)\leq \epsilon$. The achievability proof is now complete.

\subsection{Converse Coding Theorem}
We first prove a type-based ``strong converse''. Define $\overline{d}_1 :=\max_{x,y}d_1(x,y)$ and $\overline{d}_2:=\max_{x,z}d_2(x,z)$.
\begin{lemma}
\label{typestrongconverse}
Fix $\alpha>0$ and a type $Q_X\in\calP_n(\calX)$. If the excess-distortion probability satisfies 
\begin{align}
 \Pr\big(d_1(X^n,\hat{X}^n)\leq D_1,~d_2(X^n,Z^n)\leq D_2\,\big|\,X^n\in\calT_{Q_X}\big)\geq \exp(-n\alpha), \label{eqn:lb_type_str_conv}
\end{align}
then there exists a conditional distribution $Q_{YZ|X}$ such that
\begin{align}
\log M_1 &\geq nI(Q_X,Q_{Y|X})-\vartheta_n,\\
\log (M_1M_2)&\geq nI(Q_X,Q_{YZ|X})-\vartheta_n,
\end{align}

where $\vartheta_n:=|\calX|\log (n+1)+\log n+n\alpha$, 
and the expected distortions are bounded as 
\begin{align}
\mathbb{E}_{Q_X\times Q_{YZ|X}}[d_1(X,Y)]&\leq
D_1+\frac{\overline{d}_1}{n}=:D_{1,n},\\*
\mathbb{E}_{Q_X\times Q_{YZ|X}}[d_2(X,Z)]&\leq D_2+\frac{\overline{d}_2}{n}=:D_{2,n}.
\end{align}
\end{lemma}
The proof of Lemma~\ref{typestrongconverse} is given in Appendix~\ref{prooftypesc}. The proof is done in a similar manner as~\cite[Lemma 6]{watanabe2015second} and is inspired by~\cite{wei2009strong}.

Invoking Lemma~\ref{typestrongconverse} with $\alpha=\frac{\log n}{n}$, we can lower bound the excess-distortion probability for any $(n,M_1,M_2)$-code. Define $\beta_n=|\calX|\log (n+1)+2\log n$.
Define
\begin{align}
R_{1,n}&=\frac{1}{n}\log M_1+\beta_n,\\
R_{2,n}&=\frac{1}{n}\log(M_1M_2)+\beta_n.
\end{align}

\begin{lemma}
\label{lbexcessp}
For any $(n,M_1,M_2)$-code, we have
\begin{align}
\epsilon_n(D_1,D_2)\geq \Pr\left(R_{1,n}<R_Y(\hat{T}_{X^n},D_{1,n})~\mathrm{or}~R_{2,n}<\rvR(R_{1,n},D_{1,n},D_{2,n}|\hat{T}_{X^n})\right)-\frac{1}{n}.
\end{align}
\end{lemma}
The proof of Lemma~\ref{lbexcessp} is similar to~\cite[Lemma 7]{watanabe2015second} and given in Appendix~\ref{prooflbexcessp}.

Choose $\log M_1=nR_1^*+L_1\sqrt{n}+\beta_n$ and $\log (M_1 M_2)=nR_2^*+L_2\sqrt{n}+\beta_n$. Hence, $R_{i,n}=R_i^*+\frac{L_i}{\sqrt{n}}~i=1,2$. Recall that we use the shorthand $\xi_n:=\frac{\log n}{n}$. Now for $x^n$ such that $\hat{T}_{x^n}\in\calA_n(P_X)$, applying Taylor's expansion in a similar manner as~\eqref{taylor_1} and~\eqref{taylor_2}, invoking Lemma~\ref{lbexcessp} and noting that $\Pr\left(\calF\cap \calG\right)\geq \Pr(\calF)-\Pr(\calG^{\mathrm{c}})$, we obtain
\begin{align}
\nn1-\epsilon_n(D_1,D_2)
&\leq 
\Pr\left(\frac{1}{n}\sum_{i=1}^n \jmath_Y(X_i,D_1|P_X)\leq R_1^*+\frac{L_1}{\sqrt{n}}+O\left(\xi_n\right),\right.\\*
&\qquad \left.\frac{1}{n}\sum_{i=1}^n\jmath_{YZ}(X_i,R_1^*,D_1,D_2|P_{X})\leq R_2^*+\lambda^*\frac{L_1}{\sqrt{n}}+\frac{L_2}{\sqrt{n}}+O(\xi_n)\right)+\frac{1}{n}+\frac{2|\calX|}{n^2}\label{eqn:dmscon}.
\end{align}
Note that in \eqref{eqn:dmscon}, we Taylor expand $R_Y(\hat{T}_{X^n},D_{1,n})$ at the source distribution $P_X$ and distortion level $D_1$.  We also Taylor expand $\rvR(R_{1,n},D_1,D_2|\hat{T}_{X^n})$ at $(P_X,D_1,D_2)$. The residual terms when we Taylor expand with respect to the distortion levels are of the order $O(\frac{1}{n})$, which can be absorbed into $O(\xi_n)$.

The rest of converse proof can be done similarly as the achievability part in Section~\ref{secondach} by using the uni- or multi-variate  Berry-Esseen Theorem~\cite{Ben03} for Cases (i), (ii) and (iii).

\subsection{The Special case~\eqref{eqn:ones} in Corollary~\ref{srmainresult}}
\label{proofsrmain}
Recall that for successively refinable discrete memoryless source-distortion measure triplet, $\lambda^*=0$, $\nu_1^*=0$, and $\jmath_{YZ}(x_i,R_1^*,D_1,D_2|P_X)=\jmath_{Z}(x_i,D_2|P_X)$ for $R_Y(P_X,D_1)\leq R_1^*<R_Z(P_X,D_2)$. For the achievability part, invoking~\eqref{eqn:dmsach}, we obtain
\begin{align}
1-\epsilon_n(D_1,D_2)
&\geq \Pr\left(\frac{1}{n}\sum_{i=1}^n \left(\jmath_Y(X_i,D_1|P_X)- R_1^*\right)\leq \frac{L_1}{\sqrt{n}}+O\left(\xi_n\right),\right.\\
&\qquad \left.\frac{1}{n}\sum_{i=1}^n\left(\jmath_{Z}(X_i,D_2|P_{X})- R_2^*\right)\leq \frac{L_2}{\sqrt{n}}+O(\xi_n)\right)-\frac{2|\calX|}{n^2}.
\end{align}
According to the assumption in \eqref{eqn:ones} of Corollary~\ref{srmainresult}, we have $\jmath_Y(X_i,D_1|P_X)- R_1^*=\jmath_{Z}(X_i,D_2|P_{X})- R_2^*$. Given a random variable $X$ and two real numbers $a<b$, we obtain $\Pr(X<a,~X<b)=\Pr(X<a)$. Hence,
\begin{align}
1-\epsilon_n(D_1,D_2)\geq \Pr\left(\frac{1}{n}\sum_{i=1}^n \left(\jmath_Y(X_i,D_1|P_X)- R_1^*\right)\leq \frac{\min\{L_1,L_2\}}{\sqrt{n}}+O\left(\xi_n\right)\right).
\end{align}
The rest of the proof is similar to Case (i) in Section~\ref{secondach}.

Using~\eqref{eqn:dmscon}, in a similar manner as the achievability part, we complete the proof of converse part.

\section{Proof of Moderate Deviations for A DMS}
\label{proofmdc}
Consider a rate pair $(R_1^*,R_2^*)$ satisfying the conditions in Theorem \ref{mdconstant}.
\subsection{Achievability}
\label{proofmdcach}
Define
\begin{align}
\rho_{1,n}'&=\theta_1\rho_n-\frac{c_1\log n+|\calX|\log(n+1)}{n},\\
\rho_{2,n}'&=\theta_2\rho_n-c_2\frac{\log n}{n},\\
R_{i,n}'&=R_i^*+\rho_{i,n}',~i=1,2.
\end{align} 

Consider Case (i) where $R_Y(P_X,D_1)<R_1^*<\rvR(R_1^*,D_1,D_2|P_X)$ and $R_2^*=\rvR(R_1^*,D_1,D_2|P_X)$. 
Define the typical set 
\begin{align}
\calA_{n}'(P_{X})
:=\left\{Q_{X}\in\calP_n(\calX):\left\|Q_{X}-P_{X}\right\|_{1}\leq \frac{\theta\rho_n}{\sqrt{\mathrm{V}(R_1^*,D_1,D_2|P_{X})}}\right\}. \label{eqn:typ}
\end{align}
Invoking Lemma~\ref{uppexcess} with $\frac{1}{n}\log M_1=R_1^*+\theta_1\rho_n$ and $\frac{1}{n}\log(M_1M_2)=R_2^*+\theta_2\rho_n$, we obtain
\begin{align}
\epsilon_n(D_1,D_2)
&\leq \Pr\left(R_{1,n}'<R_Y(\hat{T}_{X^n},D_1)~\mathrm{or}~R_{2,n}'<\rvR(R_{1,n}',D_1,D_2|\hat{T}_{X^n})\right)\\*
\nn&\leq \Pr\left(\hat{T}_{X^n}\notin\calA_n'(P_{X})\right)+ \Pr\left(R_{1,n}'<R_Y(\hat{T}_{X^n},D_1),\hatT_{X^n}\in\calA_n'(P_{X})\right)\\*
&\qquad+\Pr\left(R_{2,n}'<\rvR(R_{1,n}',D_1,D_2|\hat{T}_{X^n}),\hatT_{X^n}\in\calA_n'(P_{X})\right)\label{eqn:dmsmdcc1}.
\end{align}
According to Weissman {\em et al.}~\cite{weissman2003inequalities}, we obtain
\begin{align}
\Pr\left(\hat{T}_{X^n}\notin\calA_n'(P_{X})\right)\leq \exp(|\calX|)\exp\left(-\frac{ n\rho_n^2\theta^2}{2\mathrm{V}(R_1^*,D_1,D_2|P_{X})}\right). \label{eqn:weiss}
\end{align}
For any $x^n$ such that $\hat{T}_{x^n}\in\calA_n'(P_{X})$, for $n$ large enough, applying Taylor's expansion, we obtain
\begin{align}
R_Y(\hat{T}_{x^n},D_1)
&=R_Y(P_X,D_1)+\sum_{x}\left(\hat{T}_{x^n}(x)-P_X(x)\right)\jmath_Y(x_i,D_1|P_X)+O\left(\|\hat{T}_{x^n}-P_X\|^2\right)\\
&=\frac{1}{n}\sum_{i=1}^n\jmath_Y(x_i,D_1|P_X)+o\left(\rho_n\right)\label{taylor1},
\end{align}
and
\begin{align}
\nn&\rvR(R_{1,n}',D_1,D_2|\hat{T}_{x^n})\\*
&=\rvR(R_1^*,D_1,D_2|P_{X})-\lambda^*\rho_{1,n}'+\sum_{x}\left(\hat{T}_{x^n}(x)-P_{X}(x)\right)\jmath_{YZ}(x_i,R_1^*,D_1,D_2|P_{X})+O\left(\rho_{1,n}'^2+\left\|\hat{T}_{x^n}-P_X\right\|^2\right)\\
&=-\lambda^*\theta_1\rho_n+\frac{1}{n}\sum_{i=1}^n\jmath_{YZ}(x_i,R_1^*,D_1,D_2|P_{X})+o\left(\rho_n\right)\label{taylor2},
\end{align}
where~\eqref{taylor2} follows because (i) according to~\eqref{expectlemma}, $\rvR(R_1^*,D_1,D_2|P_X)=\mathbb{E}[\jmath_{YZ}(X,R_1^*,R_2^*,D_1|P_X)]$; (ii) according to~\eqref{normalmdc}, we have $\frac{\log n}{n}=o(\rho_n)$, $\rho_{i,n}'^2=O(\rho_n^2)=o(\rho_n)$; (iii) since $\hat{T}_{x^n}\in\calA_n'(P_X)$, we have $O\big( \|\hat{T}_{x^n}-P_X \|^2\big)=O(\rho_n^2)=o(\rho_n)$.

For $n$ large enough, using~\eqref{taylor1} and the Chernoff bound, we obtain that for some $\gamma>0$,
\begin{align}
\Pr\left(R_{1,n}'<R_Y(\hat{T}_{X^n},D_1),\hatT_{X^n}\in\calA_n'(P_{X})\right)
&\leq \Pr\left(\frac{1}{n}\sum_{i=1}^n\jmath_Y(X_i,D_1|P_X)>R_1^*+\rho_{1,n}+o\left(\rho_n\right)\right)\leq \exp(-n\gamma)\label{wllnt1}.
\end{align}
Invoking~\eqref{taylor2}, we obtain
\begin{align}
\nn&\Pr\left(R_{2,n}'<\rvR(R_{1,n}',D_1,D_2|\hat{T}_{X^n}),\hat{T}_{X^n}\in\calA_n'(P_{X})\right)\\
&=\Pr\left(\frac{1}{n}\sum_{i=1}^n\jmath_{YZ}(X_i,R_1^*,D_1,D_2|P_{X})>R_2^*+\left(\theta \rho_n+o\left(\rho_n\right)\right)\right)\\
&=\Pr\left(\frac{1}{n}\sum_{i=1}^n\left(\jmath_{YZ}(X_i,R_1^*,D_1,D_2|P_{X})-\rvR(R_1^*,D_1,D_2|P_{X})\right)>\theta \rho_n+o\left(\rho_n\right)\right)\label{termendsec}.
\end{align}
We bound the term in~\eqref{termendsec} at the end of this section. We show that this term is of the same order as that in~\eqref{eqn:weiss}. Thus,~\eqref{eqn:dmsmdcc1} is dominated by the first and third terms as evidenced by~\eqref{eqn:weiss},~\eqref{wllnt1} and~\eqref{termendsec}. Hence, the  moderate deviations constant for Case (i) is lower bounded by ${\theta^2}/{(2\mathrm{V}(R_1^*,D_1,D_2|P_{XY}) )}$. The proof of Case (ii) is analogous to Case (i) and hence omitted. The only difference is that we define the typical set $\calA_n''(P_X)$ such that
\begin{align}
\calA_{n}''(P_{X})
:=\left\{Q_{X}\in\calP_n(\calX):\left\|Q_{X}-P_{X}\right\|_{1}\leq \frac{\theta_1\rho_n}{\sqrt{\mathrm{V}(D_1|P_{X})}}\right\}. \label{eqn:typ2}
\end{align}
The most interesting case is Case (iii) where $R_1^*=R_Y(P_X,D_1)$ and $R_2^*=\rvR(R_1^*,D_1,D_2|P_X)$. Define 
\begin{align}
\rmV_{\mathrm{max}}(R_1^*,D_1,D_2):=\max\left\{\frac{\mathrm{V}(D_1|P_{X})}{\theta_1^2},\frac{\mathrm{V}(R_1^*,D_1,D_2|P_{X})}{\theta^2}\right\}.
\end{align}
Define the typical set
\begin{align}
\calA_n'''(P_X):=\left\{Q_{X}\in\calP_n(\calX):\left\|Q_{X}-P_{X}\right\|_{1}\leq \frac{\rho_n}{\sqrt{\rmV_{\mathrm{max}}(R_1^*,D_1,D_2)}}\right\}. \label{eqn:typ3}
\end{align}
In a similar manner as Case (i) and using the union bound, we obtain
\begin{align}
\epsilon_n(D_1,D_2)
\nn&\leq \Pr\left(\hat{T}_{X^n}\notin\calA_n'''(P_{X})\right)+ \Pr\left(R_{1,n}'<R_Y(\hat{T}_{X^n},D_1),\hatT_{X^n}\in\calA_n'''(P_{X})\right)\\*
&\qquad+\Pr\left(R_{2,n}'<\rvR(R_{1,n}',D_1,D_2|\hat{T}_{X^n}),\hatT_{X^n}\in\calA_n'''(P_{X})\right)\\
\nn&\leq \exp(|\calX|)\exp\left(-\frac{ n\rho_n^2}{2\rmV_{\mathrm{max}}(R_1^*,D_1,D_2)}\right)+\Pr\left(\frac{1}{n}\sum_{i=1}^n\left(\jmath_Y(X_i,D_1|P_X)-R_Y(P_X,D_1)\right)>\rho_{1,n}+o\left(\rho_n\right)\right)\\
&\qquad+\Pr\left(\frac{1}{n}\sum_{i=1}^n\left(\jmath_{YZ}(X_i,R_1^*,D_1,D_2|P_{X})-\rvR(R_1^*,D_1,D_2|P_{X})\right)>\theta \rho_n+o\left(\rho_n\right)\right)\label{termendthird}.
\end{align}
We bound the second and third terms in~\eqref{termendthird} at the end of this section. We show that this term is of the same order as the first term in~\eqref{termendthird}. Hence, the moderate deviations constant for Case (iii) is lower bounded by 
\begin{align}
\frac{1}{2\rmV_{\mathrm{max}}(R_1^*,D_1,D_2)}=\min\left\{\frac{\theta_1^2}{2\rmV(D_1|P_X)},\frac{\theta^2}{2\rmV(R_1^*,D_1,D_2|P_X)}\right\}.
\end{align}

\subsection{Converse}
To prove the converse part, we first define 

\begin{align}
\rho_{i,n}'&:=\theta_i\rho_n+\frac{|\calX|\log (n+1)+2\log n}{n},~i=1,2.
\end{align}

In a similar manner as the proof of Lemma~\ref{lbexcessp}, we can prove
\begin{align}
\epsilon_n(D_1,D_2)
&\geq \frac{1}{2}\Pr\left(R_1^*+\rho_{1,n}'<R_Y(\hat{T}_{X^n},D_{1,n})~\mathrm{or}~R_2^*+\rho_{2,n}'<\rvR(R_1^*+\rho_{1,n}',D_{1,n},D_{2,n}|\hat{T}_{X^n})\right)\label{lbep}.
\end{align}

We first consider Case (i) where $R_Y(P_X,D_1)<R_1^*<\rvR(R_1^*,D_1,D_2|P_X)$ and $R_2^*=\rvR(R_1^*,D_1,D_2|P_X)$. 
For large $n$ and $\hatT_{x^n}\in\calA_n'(P_{X})$ (this typical set was defined in~\eqref{eqn:typ}), applying Taylor's expansion   in a similar manner as~\eqref{eqn:dmscon} and noting that $\frac{1}{n}=o(\rho_n)$, we can further lower bound~\eqref{lbep} as follows:
\begin{align}
\epsilon_n(D_1,D_2)
&\geq \frac{1}{2}\Pr\left(R_1^*+\rho_{1,n}'<R_Y(\hat{T}_{X^n},D_{1,n})~\mathrm{or}~R_2^*+\rho_{2,n}'<\rvR(R_1^*+\rho_{1,n}',D_{1,n},D_{2,n}|\hat{T}_{X^nY^n}),\hat{T}_{X^n}\in\calA_n'(P_{X})\right)\\
&\geq\nn \frac{1}{2}\max\left\{\Pr\left(\frac{1}{n}\sum_{i=1}^n\left(\jmath_{YZ}(X_i,R_1^*,D_1,D_2|P_{X})-\rvR(R_1^*,D_1,D_2|P_{XY})\right)>\theta \rho_n+o\left(\rho_n\right)\right) \right.,\\
&\qquad\left.\Pr\left(\frac{1}{n}\sum_{i=1}^n\jmath_Y(X_i,D_1|P_X)>R_1^*+\rho_{1,n}+o\left(\rho_n\right)\right)\right\}-\frac{1}{2}\Pr\left(\hat{T}_{X^n}\notin\calA_n'(P_{X})\right)\label{eqn:lbmdp1}\\
&\nn\geq \frac{1}{2}\max\left\{\Pr\left(\frac{1}{n}\sum_{i=1}^n\left(\jmath_{YZ}(X_i,R_1^*,D_1,D_2|P_{X})-\rvR(R_1^*,D_1,D_2|P_{XY})\right)>\theta \rho_n+o\left(\rho_n\right)\right),\exp(-n\gamma)\right\} \\*
&\qquad-\frac{1}{2}\Pr\left(\hat{T}_{X^n}\notin\calA_n'(P_{X})\right),\label{eqn:lower_mdp}\\
\nn&=\frac{1}{2}\Pr\left(\frac{1}{n}\sum_{i=1}^n\left(\jmath_{YZ}(X_i,R_1^*,D_1,D_2|P_{X})-\rvR(R_1^*,D_1,D_2|P_{XY})\right)>\theta \rho_n+o\left(\rho_n\right)\right)\\*
&\qquad-\frac{1}{2}\Pr\left(\hat{T}_{X^n}\notin\calA_n'(P_{X})\right),\label{eqn:lower_mdp2}
\end{align} 
where~\eqref{eqn:lbmdp1} follows from the simple  facts that $\Pr(\calF\cap\calG)\geq \Pr(\calF)-\Pr(\calG^{\mathrm{c}})$ and $\Pr(\calF\cup\calG)\geq \max\left\{\Pr(\calF),\Pr(\calG)\right\}$ for two events $\calF,\calG$;~\eqref{eqn:lower_mdp} follows from~\eqref{wllnt1};~\eqref{eqn:lower_mdp2} holds for $n$ large enough since the maximum in~\eqref{eqn:lower_mdp} is dominated by the first term which is $\exp(-n\rho_n^2\frac{\theta^2}{2\mathrm{V}(R_1^*,D_1,D_2|P_X)})$ (to be shown in \eqref{eqn:gauorder}). Note that the second term in~\eqref{eqn:lower_mdp2} is in the same order of the first term as evidenced by~\eqref{eqn:weiss}. The proof for Case (i) is now complete. Case (ii) is analogous to Case (i) and hence is omitted.

We now consider Case (iii) where $R_1^*=R_Y(P_X,D_1)$ and $R_2^*=\rvR(R_1^*,D_1,D_2|P_X)$. In a similar manner as Case (i), we obtain
\begin{align}
&\nn\epsilon_n(D_1,D_2)\\*
&\nn\geq\frac{1}{2} \max\left\{\Pr\left(\frac{1}{n}\sum_{i=1}^n\left(\jmath_Y(X_i,D_1|P_X)-R_Y(P_X,D_1)\right)>\rho_{1,n}+o\left(\rho_n\right)\right)\right.,\\*
&\left.\qquad\Pr\left(\frac{1}{n}\sum_{i=1}^n\left(\jmath_{YZ}(X_i,R_1^*,D_1,D_2|P_{X})-\rvR(R_1^*,D_1,D_2|P_{X})\right)>\theta \rho_n+o\left(\rho_n\right)\right)\right\}-\frac{1}{2}\Pr\left(\hat{T}_{X^n}\notin\calA_n'''(P_{X})\right)\label{eqn:upper_mdp},
\end{align}
where $\calA_n'''(P_{X})$ is defined in~\eqref{eqn:typ3}. Note that the second term in~\eqref{eqn:upper_mdp} is in the same order of the first term in~\eqref{termendthird}.

Invoking~\cite[Theorem 3.7.1]{dembo2009large} and the fact that $\rho_n\to 0$,  we obtain
\begin{align}
\lim_{n\to\infty}-\frac{\log \Pr\left(\frac{1}{n}\sum_{i=1}^n\left(\jmath_Y(X_i,D_1|P_X)-R_Y(P_X,D_1)\right)>\rho_{1,n}+o\left(\rho_n\right)\right)}{n\rho_n^2}=\frac{\theta_1^2}{\mathrm{V}(D_1|P_X)},
\end{align}
and
\begin{align}
\nn&\lim_{n\to\infty}-\frac{\log\Pr\left(\frac{1}{n}\sum_{i=1}^n\left(\jmath_{YZ}(X_i,R_1^*,D_1,D_2|P_{X})-\rvR(R_1^*,D_1,D_2|P_{XY})\right)>\theta \rho_n+o\left(\rho_n\right)\right)}{n\rho_n^2}\\*
&=\frac{\theta^2}{2\mathrm{V}(R_1^*,D_1,D_2|P_{X})}\label{eqn:gauorder}.
\end{align}
Note that this calculation applies to~\eqref{termendsec},~\eqref{termendthird},~\eqref{eqn:lower_mdp2} and~\eqref{eqn:upper_mdp}. The proof is now complete.

\section{A Gaussian Memoryless Source with Quadratic Distortion Measures}
\label{sec:maingms}
In this section, we consider a GMS with the quadratic distortion measures  for both $d_1$ and $d_2$. This source-distortion measure triplet is successively refinable~\cite{equitz1991successive}. We note, though, that there exist non-successively refinable continuous source-distortion measure triplets such as the symmetric mixture of Gaussians with quadratic distortion measures~\cite{chowberger}. We do not analyze this  source  here. Here, we assume that $X^n$ is i.i.d.\ where each $X_i$ is generated according to $\calN(0,\sigma^2)$. In this section, we present the second-order coding region and moderate deviations constant as well as their proofs. Note that we cannot simply evaluate the rate-dispersion functions and plug them into Corollaries \ref{srmainresult} and \ref{srmdc} because the (achievability) proofs for those results hinged on the assumption that the alphabets $\calX$, $\calY$ and $\calZ$ are finite. 

Define $\log^+(x):=\log \max\{1,x\}$. Note that for a GMS, the rate-distortion functions are 
\begin{align}
R_Y(P_X,D_1)=\frac{1}{2}\log^+\left(\frac{\sigma^2}{D_1}\right)\label{eqn:gaurated1},\\
R_Y(P_X,D_1)=\frac{1}{2}\log^+\left(\frac{\sigma^2}{D_2}\right)\label{eqn:gaurated2}.
\end{align}
Throughout this section, we consider the case where $\sigma^2>D_1>D_2>0$.

Since a GMS with the quadratic distortion measures is successively refinable, our results in this section parallel the results for successively refinable discrete memoryless source-distortion measure triplets in Section~\ref{srsource}. However, as mentioned we need to redo the proofs as the source here is continuous. Indeed, the proofs contain several novel elements such as the use of appropriately-defined Gaussian types (analogues of discrete types~\cite{csiszar2011information}).
\begin{theorem}
\label{gausecondorder}
Depending on $(R_1^*,R_2^*)$, the optimal second-order $(R_1^*,R_2^*,D_1,D_2,\epsilon)$ coding region for the GMS with the quadratic distortion measure  is as follows:
\begin{itemize}
\item Case (i): $\frac{1}{2}\log \frac{\sigma^2}{D_1}<R_1^*<\frac{1}{2}\log \frac{\sigma^2}{D_2}$ and $R_2^*=\frac{1}{2}\log \frac{\sigma^2}{D_2}$
\begin{align}
\calL(R_1^*,R_2^*,D_1,D_2,\epsilon)=\left\{(L_1,L_2):L_2\geq \sqrt{ \frac{1}{2}} \rm\rmQ^{-1}(\epsilon)\right\}.
\end{align}
\item Case (ii): $R_1^*=\frac{1}{2}\log \frac{\sigma^2}{D_1}$ and $R_2^*>\frac{1}{2}\log \frac{\sigma^2}{D_2}$
\begin{align}
\calL(R_1^*,R_2^*,D_1,D_2,\epsilon)=\left\{(L_1,L_2):L_1 \geq \sqrt{ \frac{1}{2}} \rm\rmQ^{-1}(\epsilon)\right\}.
\end{align}
\item Case (iii): $R_1^*=\frac{1}{2}\log \frac{\sigma^2}{D_1}$ and $R_2^*=\frac{1}{2}\log \frac{\sigma^2}{D_2}$
\begin{align}
\calL(R_1^*,R_2^*,D_1,D_2,\epsilon)=\left\{(L_1,L_2): \min\{L_1,L_2\}\geq \sqrt{\frac{1}{2}} \rm\rmQ^{-1}(\epsilon)\right\}\label{eqn:gautwo}.
\end{align}
\end{itemize} 
\end{theorem}
The remark for \eqref{eqn:ones} in Corollary~\ref{srmainresult} applies here. The result in~\eqref{eqn:gautwo} implies that both  excess-distortion events in~\eqref{defexcessprob} are perfectly correlated so the one consisting of the {\em smaller} second-order rate $L_i,~i=1,2$ dominates, since the first-order rates are fixed at the first-order fundamental limits $(\frac{1}{2}\log \frac{\sigma^2}{D_1},\frac{1}{2}\log \frac{\sigma^2}{D_2})$. Qualitatively the region in \eqref{eqn:gautwo} (an unbounded  rectangle) is the same as that corresponding to $D_2=0.5$ in Figure \ref{plot_srregion}.

Recall the definition of moderate deviations constant (cf.~Definition~\ref{defmdconstant}) for $\theta_1$ and $\theta_2$.
\begin{theorem}\label{gau_md}
Depending on $(R_1^*,R_2^*)$, the moderate deviations constant for the GMS with the quadratic distortion measure  is as follows:
\begin{itemize}
\item Case (i): $\frac{1}{2}\log \frac{\sigma^2}{D_1}<R_1^*<\frac{1}{2}\log\frac{\sigma^2}{D_2}$ and $R_2^*=\frac{1}{2}\log \frac{\sigma^2}{D_2}$
\begin{align}
\nu^*(R_1^*,R_2^*|D_1,D_2)=\theta_2^2.
\end{align}
\item Case (ii): $R_1^*=\frac{1}{2}\log \frac{\sigma^2}{D_1}$ and $R_2^*>\frac{1}{2}\log \frac{\sigma^2}{D_2}$
\begin{align}
\nu^*(R_1^*,R_2^*|D_1,D_2)=\theta_1^2.
\end{align}
\item Case (iii): $R_1^*=\frac{1}{2}\log \frac{\sigma^2}{D_1}$ and $R_2^*=\frac{1}{2}\log \frac{\sigma^2}{D_2}$
\begin{align}
\nu^*(R_1^*,R_2^*|D_1,D_2)=\min\{\theta_1^2,\theta^2_2\}.
\end{align}
\end{itemize} 
\end{theorem}
We note that the results in Theorem \ref{gausecondorder} and \ref{gau_md} do not depend on the distortion and the source variance. This is expected since the dispersion of lossy source coding for Gaussian sources is $1/2$ nats$^2$ per source symbol~\cite{ingber2011,kostina2012fixed}. Similarly the moderate deviations constant for Gaussian rate-distortion also does not depend on the distortion level and the source variance~\cite{tan2012moderate}.

\subsection{Preliminaries for the  Proofs}
In this subsection, we present some preliminaries for the proofs of Theorems \ref{gausecondorder} and \ref{gau_md}. In particular, we present an appropriate definition of Gaussian types for our problem and  a type covering lemma for Gaussian types.

Let $\xi>0$ be specified later. Define the typical set 
\begin{align}
\calU^{\xi}:=\left\{x^n: e^{-2\xi}<\frac{\|x^n\|^2}{n\sigma^2}<e^{2\xi}\right\}.
\end{align}
In a similar manner as Eqn.~(35) in~\cite{tan2012moderate} (Cram\'er's theorem~\cite{dembo2009large}), we obtain
\begin{align}
\Pr\left(X^n\notin \calU^{\xi}\right)
&\leq 4\exp\left(-n I(\xi) \right) \label{eqn:cramer} 
\end{align} 
where the {\em large deviations rate function} of the $\chi_1^2$ random variable  is  
\begin{equation}
I(\xi) := \frac{1}{2} (e^{2\xi}-1-2\xi). \label{eqn:ratefunc}
\end{equation}
Let $\delta>0$ be specified later and let the number of types be
\begin{equation}
k=\left\lceil{\frac{e^{2\xi}-e^{-2\xi}}{\delta}} \right\rceil+1. \label{eqn:num_types}
\end{equation}
Note that $\delta $ and $\xi$ control the number of types. Define $\Lambda(i)=\sigma^2e^{-2\xi}+(i-1)\delta\sigma^2$. Also define the {\em GMS type classes} 
\begin{align}
\calU_i:=\left\{x^n: \Lambda(i-1)\leq \frac{\|x^n\|^2}{n}\leq \Lambda(i)\right\},~i\in[1:k]\label{defui}.
\end{align}
Hence,
\begin{align}
\calU^{\xi}\subset\bigcup_{i=1}^k \calU_i.
\end{align}
Note that $\calU_i$ is a collection of GMS sequences with normalized squared $l_2$ norm (power) within $(\Lambda(i-1),\Lambda(i)]$. Hence, we define the type of a GMS sequence $x^n$ as $i$ if $x^n\in\calU_i$. In particular, if $x^n\notin \calU_{i}$ for all $i\in[1:k]$, we define the type of $x^n$ as $0$. See also ~\cite[Eqn.~(61)]{arikan1998guessing} and~\cite[Definition~1]{kelly2012reliability} for other definitions of Gaussian type classes.

We then present a type covering lemma for a  GMS with the quadratic distortion measures which is analogous to the type covering lemma for a DMS with arbitrary distortion measures in Lemma~\ref{typecovering}.
\begin{lemma}
\label{Gaucovering}
Given $x^n\in\calU_i$, the following holds:
\begin{itemize}
\item There exists a set $\calB_Y\subset\bbR^n$ such that 
\begin{align}
\log |\calB_Y|\leq \frac{n}{2}\log \frac{\Lambda(i)}{D_1}+\frac{5}{2}\log n+\log 6,
\end{align}
and $\calB_Y$ $D_1$-covers $\calU_i$, i.e.,
\begin{align}
\calU_i\subset \bigcup_{y^n\in\calB_Y}\calN(y^n,D_1),
\end{align}
where 
\begin{align}
\calN(y^n,D_1):=\left\{x^n:\|x^n-y^n\|^2\leq D_1\right\}.
\end{align}
\item For each $y^n\in\calB_Y$, there exists a set $\calB_{Z}(y^n)\subset\bbR^n$ such that 
\begin{align}
\log \left(\sum_{y^n\in\calB_Y}|\calB_Z(y^n)|\right)\leq \frac{n}{2}\log \frac{\Lambda(i)}{D_2}+5\log n+2\log 6,
\end{align}
and $\calB_{Z}(y^n)$ $D_2$-covers $\calN_1(y^n,D_1)$, i.e.,
\begin{align}
\calN_1(y^n,D_1)=\bigcup_{z^n\in\calB_Z(y^n)}\calN(z^n,D_2).
\end{align}
\end{itemize}
\end{lemma}
The proof of Lemma~\ref{Gaucovering} uses~\cite[Theorem 1.2]{verger2005covering} multiple times. For the first reconstruction using $Y^n$, we observe that $6n^{5/2}(\Lambda(i)/D_1)^{n/2}$ points can $D_1$-cover $\calU_i$. For the second reconstruction using $Z^n$, we observe that  $6n^{5/2}(D_1/D_2)^{n/2}$ points suffice to $D_2$-cover each ball centered at $y^n\in\calB_{Y}$ with radius $D_1$.

We now present an upper bound on the excess-distortion probability of  the code prescribed by Lemma \ref{Gaucovering}. Recall that $k$ is the number of types. Similarly to the proof of Lemma~\ref{uppexcess} for a DMS, for a GMS, we also need to transmit the type. This requires no more that $\log k$ nats.  Observe that there is a tradeoff between the size of the typical set controlled by $\xi$ and the number of types $k$. As $\xi$ increases, the probability that a sequence is atypical decreases. See \eqref{eqn:cramer}--\eqref{eqn:ratefunc}. However, the number of types increases. Depending on the regime (second-order or moderate deviations) we will choose $\xi$ differently. Now, given any $(n,M_1,M_2)$-code, define
\begin{align}
nR_{1,n}&=\log M_1-\frac{5}{2}\log n-\log k-\log 6, \label{eqn:defR1n}\\*
nR_{2,n}&=\log (M_1M_2)-5\log n-2\log 6.
\end{align}
\begin{lemma}
\label{gauuppexcess}
There exists an $(n,M_1,M_2)$-code such that 
\begin{align}&\epsilon_n(D_1,D_2)\leq 4\exp\left(-nI(\xi)\right)+\Pr\left(\frac{1}{n}\sum_{i=1}^n \frac{X_i^2}{\sigma^2}>\frac{D_1}{\sigma^2}\exp(2R_{1,n})-\delta~\mathrm{or}~\frac{1}{n}\sum_{i=1}^n \frac{X_i^2}{\sigma^2}>\frac{D_2}{\sigma^2}\exp(2R_{2,n})-\delta\right)\label{gautypeach}.
\end{align}
\end{lemma}
The proof of Lemma~\ref{gauuppexcess} is analogous to Lemma~\ref{uppexcess} and given in Appendix~\ref{proofgauuppexcess}.

\subsection{Proof of Second-Order Asymptotics (Theorem~\ref{gausecondorder})} \label{sec:prf_g_2}
\label{proofgausecond}
We begin with the achievability for Theorem \ref{gausecondorder}. Let $\xi :=n^{-1/3}$ and $\delta :=1/n$. Invoking~\eqref{eqn:ratefunc} and Taylor expansion, we obtain that the first term on the right-hand-side of \eqref{gautypeach} behaves as
\begin{align}
4\exp\left(-n I(\xi) \right)=4\exp\left(-n^{1/3}+o(n^{1/3})\right) =:\kappa_n\to 0.
\end{align}
Additionally, define
\begin{align}
\log M_1&=nR_1^*+L_1\sqrt{n}+\frac{7}{2}\log n+\log 6, \label{eqn:defM1}\\*
\log (M_1M_2)&=nR_2^*+L_2\sqrt{n}+5\log n+2\log 6.\label{eqn:defM2}
\end{align}
Now, with our choice of $\xi$ and $\delta$, we see from \eqref{eqn:num_types} that the number of types is $k=\ceil{4n^{2/3}+O(n^{1/3})}+1$. For $n$ large enough, $k\leq n$. Hence, we only have \emph{polynomially} (in fact at most {\em linearly}) many types. Furthermore, observe that the coefficient of the $\log n$ terms in \eqref{eqn:defR1n} and \eqref{eqn:defM1} differ by one because  we need to transmit the type requiring $\log k\le \log n$ nats (cf.\ proof of Lemma \ref{gauuppexcess}). The terms scaling as $O(\log n)$ in \eqref{eqn:defM1} and \eqref{eqn:defM2} do not affect the second-order coding region.

Now, note that $Y_i=X_i^2/\sigma^2$ and $Y_i$ is $\chi_1^2$-distributed.  Invoking Lemma~\ref{gauuppexcess}, we obtain
\begin{align}
\nn&1-\epsilon_n(D_1,D_2)\\*
&\geq -\kappa_n+\Pr\left(\frac{1}{n}\sum_{i=1}^n Y_i\leq \frac{D_1}{\sigma^2}\exp\left(2\left(R_1^*+\frac{L_1}{\sqrt{n}}\right)\right)-\frac{1}{n},~\frac{1}{n}\sum_{i=1}^n Y_i\leq \frac{D_2}{\sigma^2}\exp\left(2\left(R_2^*+\frac{L_2}{\sqrt{n}}\right)\right)-\frac{1}{n}\right)\label{eqn:gaulbexcess}.
\end{align}
We now consider different cases. We first consider Case (i) where $\frac{1}{2}\log \frac{\sigma^2}{D_1}<R_1^*<\frac{1}{2}\log \frac{\sigma^2}{D_2}$ and $R_2^*=\frac{1}{2}\log \frac{\sigma^2}{D_2}$. Since $e^x\geq 1+x$ for all $x \in\bbR$, we obtain
\begin{align}
\exp\left(\frac{2L_i}{\sqrt{n}}\right)
&\geq 1+\frac{2L_i}{\sqrt{n}},~i=1,2.
\end{align}
Define $\tau_n=\frac{1}{n}$. According to the weak law of large numbers, we obtain
\begin{align}
\nn&\Pr\left(\frac{1}{n}\sum_{i=1}^n Y_i\leq \frac{D_1}{\sigma^2}\exp\left(2\left(R_1^*+\frac{L_1}{\sqrt{n}}\right)\right)-\tau_n\right)\\
&\geq \Pr\left(\frac{1}{n}\sum_{i=1}^n Y_i\leq \frac{D_1}{\sigma^2}\exp(2R_1^*)\left(1+\frac{2L_1}{\sqrt{n}}\right)-\tau_n\right)\to 1.
\end{align}
Invoking  the Berry-Esseen Theorem, we obtain
\begin{align}
\nn&\Pr\left(\frac{1}{n}\sum_{i=1}^n Y_i\leq \frac{D_2}{\sigma^2} \exp\left(2\left(\frac{1}{2}\log \frac{\sigma^2}{D_2}+\frac{L_2}{\sqrt{n}}\right)\right)-\tau_n\right)\\
&\geq \Pr\left(\frac{1}{n}\sum_{i=1}^n Y_i\leq 1+\frac{2L_2}{\sqrt{n}}-\tau_n\right)\\
&\geq 1-\rmQ\left(\sqrt{2}L_2-\frac{1}{\sqrt{n}}\right)-\sqrt{\frac{2}{n}}\label{deftaun}.
\end{align}
Hence, by using the bound in \eqref{eqn:gaulbexcess}, we obtain
\begin{align}
\epsilon_n(D_1,D_2)\leq \rmQ\left(\sqrt{2}L_2-\frac{1}{\sqrt{n}}\right)+\kappa_n+\sqrt{\frac{2}{n}}
\end{align}
Thus, if $(L_1,L_2)$ satisfy 
\begin{align}
L_2\geq\sqrt{\frac{1}{2}} \rmQ^{-1}(\epsilon),
\end{align}
then $\limsup_{n\to\infty}\epsilon_n(D_1,D_2)\leq \epsilon$.

Case (ii) is analogous to Case (i) and thus omitted. The most interesting case is Case (iii), where $R_1^*=\frac{1}{2}\log \frac{\sigma^2}{D_1}$ and $R_2^*=\frac{1}{2}\log \frac{\sigma^2}{D_2}$. The covariance matrix of $[Y_i-1, Y_i-1]$ is
\begin{align}
\mathbf{V}=2\cdot \mathrm{ones}(2,2).
\end{align}
Because $\bV$ is singular, we cannot use the multi-variate Berry-Esseen Theorem here. However, the analysis is simple. Indeed, 
\begin{align}
&\nn\Pr\left(\frac{1}{n}\sum_{i=1}^n Y_i\leq \frac{D_1}{\sigma^2}\exp\left(2\left(R_1^*+\frac{L_1}{\sqrt{n}}\right)\right)-\tau_n,~\frac{1}{n}\sum_{i=1}^n Y_i\leq \frac{D_2}{\sigma^2}\exp\left(2\left(R_2^*+\frac{L_2}{\sqrt{n}}\right)\right)-\tau_n\right)\\
&\geq \Pr\left(\frac{1}{n}\sum_{i=1}^n (Y_i-1)\leq \frac{2L_1}{\sqrt{n}}-\tau_n,~\frac{1}{n}\sum_{i=1}^n (Y_i-1)\leq \frac{2L_2}{\sqrt{n}}-\tau_n\right)\\
&=\Pr\left(\frac{1}{n}\sum_{i=1}^n (Y_i-1)\leq \frac{2\min\{L_1,L_2\}}{\sqrt{n}}-\tau_n\right)\\*
&\geq 1-\rmQ\left(\sqrt{2}\min\{L_1,L_2\}-\frac{1}{\sqrt{n}}\right)-\sqrt{\frac{2}{n}}.
\end{align}
Hence,
\begin{align}
\epsilon_n(D_1,D_2)\leq \rmQ\left(\sqrt{2}\min\{L_1,L_2\}-\frac{1}{\sqrt{n}}\right)+\kappa_n+\sqrt{\frac{2}{n}}
\end{align}
If $(L_1,L_2)$ satisfy 
\begin{align}
\min\{L_1,L_2\}\geq \sqrt{\frac{1}{2}} \rmQ^{-1}(\epsilon),
\end{align}
then $\limsup_{n\to\infty}\epsilon_n(D_1,D_2)\leq \epsilon$.

Next, we turn to the converse proof. This follows from Lemma~\ref{tiltedconverse}. As shown in~\cite[Example 2]{kostina2012fixed}, for a GMS $\calN(0,\sigma^2)$, 
\begin{align}
\jmath_Y(x,D_1|P_X)&=\frac{1}{2}\log \frac{\sigma^2}{D_1}+\frac{1}{2}\left(\frac{x^2}{\sigma^2}-1\right),
\end{align}
and similarly for $\jmath_Z(x,D_2|P_X)$. 
Hence,
\begin{align}
\mathrm{V}(D_1|P_X)=\mathrm{V}(D_2|P_X)= \frac{1}{2},\label{eqn:vard1}
\end{align}
and similarly,
\begin{align}
\mathrm{Cov}[\jmath_Y(X,D_1|P_X),\jmath_Z(X,D_2|P_X)]
&=\frac{1}{2}.
\end{align}
The covariance matrix is
\begin{align}
\mathbf{V}(D_1,D_2|P_X)=\frac{1}{2} \cdot \mathrm{ones}(2,2) .
\end{align}
The rest of the proof is similar to  the converse proof of Case (iii)(a) of Corollary~\ref{srmainresult} (Section~\ref{proofsrmain}).

\subsection{Proof of Moderate Deviations (Theorem~\ref{gau_md})}\label{sec:prf_g_m}
The achievability part can be done in a similar manner as~\cite[Theorem 5]{tan2012moderate}. Here we provide an alternative proof which parallels our analysis for a DMS in Section~\ref{proofmdcach} and the achievability proof of second-order asymptotics for the a GMS in Section~\ref{proofgausecond}.

Define 
\begin{align}
\rho_{1,n}'&:=\theta_1\rho_n-\frac{7\log n}{2n}-\frac{\log 6}{2n},\\
\rho_{2,n}'&:=\theta_2\rho_n-\frac{5\log n}{n}-\frac{2\log 6}{2n}.
\end{align}

We first consider Case (i) where $\frac{1}{2}\log \frac{\sigma^2}{D_1}<R_1^*<\frac{1}{2}\log\frac{\sigma^2}{D_2}$ and $R_2^*=\frac{1}{2}\log\frac{\sigma^2}{D_2}$. Choose $\xi^2=\frac{\theta_2^2\rho_n^2}{2\mathrm{V}(D_2|P_X)}$ and $\delta=\frac{1}{n}$. Then  from \eqref{eqn:num_types}, $k=\ceil{n(4\xi+O(\xi^2))}+1=O(n\rho_n)\leq n$ for large $n$. Thus, similarly to the proof of the achievability part for the second-order asymptotics, we have only \emph{at most linearly} many types which requires $\log k\leq \log n$ nats to transmit and does not affect the moderate deviations constant. Thus, invoking Lemma~\ref{gauuppexcess}, we see that there exists an $(n,M_1,M_2)$-code such that 
\begin{align}
\log M_i=n\left(R_i^*+\rho_{i,n}'\right)~i=1,2,
\end{align}
and
\begin{align}
&\epsilon_n(D_1,D_2)\nn\leq 
4\exp\left(-n\left(\frac{\theta_2^2\rho_n^2}{2\mathrm{V}(D_2|P_X)}+o(\rho_n^2)\right)\right)\\*
&\qquad+\Pr\left(\frac{1}{n}\sum_{i=1}^n \frac{X_i^2}{\sigma^2}>\frac{D_1}{\sigma^2}\exp(2(R_1^*+\rho_{1,n}'))-\frac{1}{n}~\mathrm{or}~\frac{1}{n}\sum_{i=1}^n \frac{X_i^2}{\sigma^2}>\frac{D_2}{\sigma^2}\exp(2(R_2^*+\rho_{2,n}'))-\frac{1}{n}\right)\label{eqn:gaumdcup1}.
\end{align}
We then focus on the second term in~\eqref{eqn:gaumdcup1}. According to the Chernoff bound, we obtain that for some constant $\gamma>0$,
\begin{align}
\Pr\left(\frac{1}{n}\sum_{i=1}^n \frac{X_i^2}{\sigma^2}>\frac{D_1}{\sigma^2}\exp(2(R_1^*+\rho_{1,n}'))-\frac{1}{n}\right)\leq \exp(-n\gamma).
\end{align}
By the union bound,
\begin{align}
&\nn\epsilon_n(D_1,D_2)\\*
&\leq 
4\exp\left(-n\left(\frac{\theta_2^2\rho_n^2}{2\mathrm{V}(D_2|P_X)}+o(\rho_n^2)\right)\right)+\exp(-n\gamma)+\Pr\left(\frac{1}{n}\sum_{i=1}^n \frac{X_i^2}{\sigma^2}>\frac{D_2}{\sigma^2}\exp(2(R_2^*+\rho_{2,n}'))-\frac{1}{n}\right)\label{eqn:unimdcgau}.
\end{align}
Recall that $Y_i=\frac{X_i^2}{\sigma^2}$. For the third term in~\eqref{eqn:unimdcgau},
\begin{align}
\Pr\left(\frac{1}{n}\sum_{i=1}^n \frac{X_i^2}{\sigma^2}>\frac{D_2}{\sigma^2}\exp(2(R_2^*+\rho_{2,n}'))-\frac{1}{n}\right)
&=\Pr\left(\frac{1}{n}\sum_{i=1}^n (Y_i-1)>2\rho_{2,n}'+O(\rho_{2,n}'^2)-\frac{1}{n}\right)\label{eqn:taylorgauc1}\\
&=\Pr\left(\frac{1}{n}\sum_{i=1}^n (Y_i-1)>2\theta_2\rho_n+o(\rho_n)\right)\label{eqn:proprhon}.
\end{align}
where~\eqref{eqn:taylorgauc1} follows from Taylor expansion and~\eqref{eqn:proprhon} follows due to~\eqref{normalmdc}, from which we have $\frac{\log n}{n}=o(\rho_n)$, $\frac{1}{n}=o(\rho_n)$, and $\rho_{2,n}'^2 =O(\rho_n^2)=o(\rho_n)$. Invoking~\cite[Theorem 3.7.1]{dembo2009large}, we obtain
\begin{align}
\lim_{n\to\infty}-\frac{\log \Pr\left(\frac{1}{n}\sum_{i=1}^n (Y_i-1)>2\theta_2\rho_n+o(\rho_n)\right)}{n\rho_n^2}=\frac{4\theta_2^2}{4}=\theta_2^2\label{eqn:gaumd2}.
\end{align}
Note that in~\eqref{eqn:unimdcgau}, the first and third terms dominate  and they decay at the same rate (cf.~\eqref{eqn:vard1}). Hence, we obtain
\begin{align}
\liminf_{n\to\infty}-\frac{\log \epsilon_n(D_1,D_2)}{n\rho_n^2}\geq \theta_2^2.
\end{align}
Case (ii) is analogous to Case (i) and thus omitted. We thus focus on case (iii) where $R_1^*=\frac{1}{2}\log\frac{\sigma^2}{D_1}$ and $R_2^*=\frac{1}{2}\log\frac{\sigma^2}{D_2}$. Choose $\delta=\frac{1}{n}$ and
\begin{align}
\xi^2=\min\left\{\frac{\theta_1^2\rho_n^2}{2\mathrm{V}(D_1|P_X)},\frac{\theta_1^2\rho_n^2}{2\mathrm{V}(D_2|P_X)}\right\}.
\end{align}
In a similar manner as Case (i), we can prove that  that there exists an $(n,M_1,M_2)$-code such that 
\begin{align}
\log M_i=n\left(R_i^*+\rho_{i,n}'\right)~i=1,2,
\end{align}
and
\begin{align}
&\epsilon_n(D_1,D_2)\nn\leq 
4\exp\left(-n\left(\frac{\theta_2^2\rho_n^2}{2\mathrm{V}(D_2|P_X)}+o(\rho_n^2)\right)\right)\\
&\qquad+\Pr\left(\frac{1}{n}\sum_{i=1}^n \frac{X_i^2}{\sigma^2}>\frac{D_1}{\sigma^2}\exp(2(R_1^*+\rho_{1,n}'))-\frac{1}{n}~\mathrm{or}~\frac{1}{n}\sum_{i=1}^n \frac{X_i^2}{\sigma^2}>\frac{D_2}{\sigma^2}\exp(2(R_2^*+\rho_{2,n}'))-\frac{1}{n}\right)\label{eqn:gaumdcup2}.
\end{align}
Denote the second term in~\eqref{eqn:gaumdcup2} as $\epsilon_n^{\mathrm{(ii)}}(D_1,D_2)$. In a similar manner as Case (i) and using the union bound, we obtain
\begin{align}
\epsilon_n^{\mathrm{(ii)}}(D_1,D_2)\leq 
\Pr\left(\frac{1}{n}\sum_{i=1}^n \left(Y_i-1\right)>2\theta_1\rho_n+o(\rho_n)\right)+\Pr\left(\frac{1}{n}\sum_{i=1}^n \left(Y_i-1\right)>2\theta_2\rho_n+o(\rho_n)\right).
\end{align}

Invoking~\cite[Theorem 3.7.1]{dembo2009large}, we obtain
\begin{align}
\lim_{n\to\infty}-\frac{\log \Pr\left(\frac{1}{n}\sum_{i=1}^n (Y_i-1)>2\theta_1\rho_n+o(\rho_n)\right)}{n\rho_n^2}=\frac{4\theta_1^2}{4}=\theta_1^2\label{eqn:gaumd1}.
\end{align}
Hence, combining~\eqref{eqn:gaumdcup2} with~\eqref{eqn:vard1},~\eqref{eqn:gaumd2} and~\eqref{eqn:gaumd1}, we obtain
\begin{align}
\liminf_{n\to\infty}-\frac{\log \epsilon_n(D_1,D_2)}{n\rho_n^2}\geq \min\{\theta_1^2,\theta_2^2\}.
\end{align}

The converse part follows from Lemma~\ref{tiltedconverse}.
Let $\zeta\in(0,1)$ be arbitrary. For $i=1,2$, let $\gamma_i=n\zeta\theta_i\rho_{n}$ and $\log M_i=n(R_i^*+\theta_i\rho_n)$. Using~\eqref{eqn:tiltex1}, we obtain
\begin{align}
\nn\epsilon_n(D_1,D_2)
&\geq 
\max\left\{\Pr\left(\sum_{i=1}^n\jmath_Y(X_i,D_1|P_X)\geq nR_1^*+n(1+\zeta)\theta_1\rho_n\right),\Pr\left(\sum_{i=1}^n\jmath_Z(X_i,D_2|P_X)\geq nR_2^*+n(1+\zeta)\theta_2\rho_n\right)\right\} \\*
&\qquad-\exp(-n\zeta\theta_1\rho_n)-\exp(-n\zeta\theta_2\rho_n).
\end{align}

Now we consider the different cases. We first consider Case (i) where $R_Y(P_X,D_1)<R_1^*<R_{Z}(P_X,D_2)$ and $R_2^*=R_{Z}(P_X,D_2)$. From the Chernoff bound, we obtain for some $\gamma>0$,
\begin{align}
\Pr\left(\sum_{i=1}^n\jmath_Y(X_i,D_1|P_X)\geq nR_1^*+n(1+\zeta)\theta_1\rho_n\right)\leq \exp(-n\gamma).
\end{align}
Hence, we obtain
\begin{align}
\nn&\epsilon_n(D_1,D_2)\\*
&\geq 
\max\left\{\Pr\left(\frac{1}{n}\sum_{i=1}^n\left(\jmath_Z(X_i,D_2|P_X)- R_2^*\right)\geq (1+\zeta)\theta_2\rho_n\right),\exp(-n\gamma)\right\}-\exp(-n\zeta\theta_1\rho_n)-\exp(-n\zeta\theta_2\rho_n)\label{gaumdc_1}.
\end{align}
Using~\cite[Theorem 3.7.1]{dembo2009large}, we obtain that 
\begin{align}
\lim_{n\to\infty}-\frac{\log \Pr\left(\frac{1}{n}\sum_{i=1}^n\left(\jmath_Z(X_i,D_2|P_X)- R_2^*\right)\geq (1+\zeta)\theta_2\rho_n\right)}{n\rho_n^2}=\frac{(1+\zeta)^2\theta_2^2}{2\mathrm{V}(D_2|P_X)}.
\end{align}
Hence, for large $n$, due to the fact that  $\exp\left(-n\rho_n^2\frac{(1+\zeta)^2\theta_2^2}{2\mathrm{V}(D_2|P_X)}\right)$ dominates $\exp(-n\gamma)$, we obtain
\begin{align}
\epsilon_n(D_1,D_2)
&\geq \Pr\left(\frac{1}{n}\sum_{i=1}^n\left(\jmath_Z(X_i,D_2|P_X)- R_2^*\right)\geq (1+\zeta)\theta_2\rho_n\right)-\exp\left(-n\zeta\theta_1\rho_n)-\exp(-n\zeta\theta_2\rho_n\right)\label{gaumdc1}.
\end{align}
Note that for large $n$,~\eqref{gaumdc1} is dominated by the first term since $\exp\left(-n\rho_n^2\frac{(1+\zeta)^2\theta_2^2}{2\mathrm{V}(D_2|P_X)}\right)$ dominates $\exp(-n\zeta\theta_i\rho_n),~i=1,2$. Thus, 
\begin{align}
\limsup_{n\to\infty}-\frac{\epsilon_n(D_1,D_2)}{n\rho_n^2}\leq \frac{(1+\zeta)^2\theta_2^2}{2\mathrm{V}(D_2|P_X)}.
\end{align} 
Case (ii) is analogous to Case (i) and thus omitted. We now consider Case (iii) where $R_1^*=R_Y(P_X,D_1)$ and $R_2^*=R_{Z}(P_X,D_2)$. In a similar manner as Case (i), we can prove that
\begin{align}
\nn\epsilon_n(D_1,D_2)
&\geq 
\max\left\{\Pr\left(\frac{1}{n}\sum_{i=1}^n\left(\jmath_Y(X_i,D_1|P_X)-R_1^*\right)\geq (1+\zeta)\theta_1\rho_n\right),\Pr\left(\frac{1}{n}\sum_{i=1}^n\left(\jmath_Z(X_i,D_2|P_X)-R_2^*\right)\geq (1+\zeta)\theta_2\rho_n\right)\right\} \\
&\qquad-\exp(-n\zeta\theta_1\rho_n)-\exp(-n\zeta\theta_2\rho_n)\label{gaumdc3}.
\end{align}
Invoking~\cite[Theorem 3.7.1]{dembo2009large} again, we obtain that
\begin{align}
\lim_{n\to\infty}-\frac{\log \Pr\left(\frac{1}{n}\sum_{i=1}^n\left(\jmath_Y(X_i,D_2|P_X)- R_1^*\right)\geq (1+\zeta)\theta_1\rho_n\right)}{n\rho_n^2}=\frac{(1+\zeta)^2\theta_1^2}{2\mathrm{V}(D_1|P_X)}.
\end{align}
Note that~\eqref{gaumdc3} is dominated by the first term. Thus,
\begin{align}
\limsup_{n\to\infty}-\frac{\epsilon_n(D_1,D_2)}{n\rho_n^2}\leq \min\left\{\frac{(1+\zeta)^2\theta_1^2}{2\mathrm{V}(D_1|P_X)},\frac{(1+\zeta)^2\theta_2^2}{2\mathrm{V}(D_2|P_X)}\right\}.
\end{align}
For all cases, let $\zeta\downarrow 0$. This completes the proof for a GMS by appealing to~\eqref{eqn:vard1}.

\section{Conclusion}
\label{sec:conc}
In this paper, we have derived the second-order coding region and moderate deviations constant for the successive refinement source coding problem under \emph{joint} excess distortion event. We did so for both a DMS with arbitrary distortion measures and a GMS with the quadratic distortion measures and obtained  some new insights. Our results for a DMS with arbitrary distortion measures can be specialized to successively refinable discrete memoryless source-distortion measure triplets to obtain simpler expressions. 

In the future, one may derive the second-order asymptotics and moderate deviations for a Laplacian source  with the absolute distortion measures \cite{zhong2006type} following a similar method as used in this paper. Since this source-distortion measure triplet is successively refinable~\cite{equitz1991successive}, we do not envision any significant difficulties. We may also endeavor to do the same for more challenging source-distortion measure triplets such as the symmetric mixture of Gaussians~\cite{chowberger}, which is a continuous source that is not successively refinable, hence new techniques may be required. We also aim to derive the second-order asymptotics and moderate deviations constant for the multiple description source coding problem with one deterministic decoder~\cite{fu2002rate}. This may be done, possibly, using similar methods to those introduced in this paper.

\appendix
\subsection{Proof of Lemma~\ref{propertytilted}}
\label{proofpropertytilted}
It is easy to observe that $\rvR(R_1^*,D_1,D_2|P_X)$ in~\eqref{minimumr2} is a convex optimization problem. For $(\lambda,\nu_1,\nu_2)\in\mathbb{R}_+^3$, define
\begin{align}
g(\lambda,\nu_1,\nu_2):=\inf_{P_{YZ|X}} I(X;YZ)+\lambda(I(X;Y_1)-R_1^*)+\nu_1(\mathbb{E}[d_1(X,Y)]-D_1)+\nu_2(\mathbb{E}[d_2(X,Z)]-D_2).
\end{align}
Considering the dual problem, we obtain
\begin{align}
\label{convexopt}
\rvR(R_1^*,D_1,D_2|P_X)=\max_{(\lambda,\nu_1,\nu_2)\in\mathbb{R}_+^3}g(\lambda,\nu_1,\nu_2)=g(\lambda^*,\nu_1^*,\nu_2^*).
\end{align}
Given $Q_{YZ}$, let $Q_Y$ be the induced marginal distribution on $\calY$. For arbitrary $P_{YZ|X}$ and $Q_{YZ}$, define
\begin{align}
\nn&F(P_{YZ|X},Q_{YZ}|R_1^*,D_1,D_2)\\*
&:=D(P_{YZ|X}\|Q_{YZ}|P_X)+\lambda^*(D(P_{Y|X}\|Q_Y|P_X)-R_1^*)+\nu_1^*(\mathbb{E}[d_1(X,Y)]-D_1)+\nu_2^*(\mathbb{E}[d_2(X,Z)]-D_2)\\
&=I(X;YZ)+D(P_{YZ}\|Q_{YZ})+\lambda^*(I(X;Y)+D(P_Y\|Q_Y)-R_1^*)\nn\\*
&\qquad +\nu_1^*(\mathbb{E}[d_1(X,Y)]-D_1)+\nu_2^*(\mathbb{E}[d_2(X,Z)]-D_2).
\end{align}

For $(\lambda,\nu_1,\nu_2)\in\mathbb{R}_+^3$ and $(x,y,z)\in\calX\times\calY\times\calZ$, define 
\begin{align}
f(x,y,z|Q_{YZ},\lambda,\nu_1,\nu_2):=\exp\left(-\lambda\left(\log\frac{P_{Y|X}(y|x)}{Q_Y(y)}-R_1^*\right)-\nu_1(d_1(x,y)-D_1)-\nu_2(d_2(x,z)-D_2)\right)
\end{align}
Then we can define the generalized tilted information density
\begin{align}
\Lambda(x|Q_{YZ},\lambda,\nu_1,\nu_2):=-\log\mathbb{E}_{Q_{YZ}}\left[f(x,Y,Z|Q_{YZ},\lambda,\nu_1,\nu_2)\right].
\end{align}
We can relate $F(P_{YZ|X},Q_{YZ}|R_1^*,D_1,D_2)$ and $\Lambda(x|Q_{YZ},Q_{Y},\lambda,\nu_1,\nu_2)$ in the following lemma.  
\begin{lemma}
\label{relatefl}
\begin{align}
F(P_{YZ|X},Q_{YZ}|R_1^*,D_1,D_2)\geq \mathbb{E}_{P_X}\left[\Lambda(X|Q_{YZ},\lambda^*,\nu_1^*,\nu_2^*)\right],
\end{align}
with equality if and only if $P_{YZ|X}$ satisfies
\begin{align}
P_{YZ|X}(yz|x)&=Q_{YZ}(y,z)\exp\left(\Lambda(x|Q_{YZ},Q_{Y},\lambda^*,\nu_1^*,\nu_2^*)-\lambda^*\left(\log\frac{P_{Y|X}(y|x)}{Q_Y(y)}-R_1^*\right)\right.\\*
&\qquad\left.-\nu_1^*(d_1(x,y)-D_1)-\nu_2^*(d_2(x,z)-D_2)\vphantom{\frac{P_{Y|X}(y|x)}{Q_Y(y)}}\right).
\end{align}
\end{lemma}
\begin{proof}
Invoking the log-sum inequality, we obtain
\begin{align}
\nn&F(P_{YZ|X},Q_{YZ}|R_1^*,D_1,D_2)\\*
\nn&=\sum_{x,y,z}P_{X}(x)P_{YZ|X}(y,z|x)\left(\log\frac{P_{YZ|X}(y,z|x)}{Q_{YZ}(y,z)}+\lambda^*\left(\log\frac{P_{Y|X}(y|x)}{Q_{Y}(y)}-R_1^*\right)\right.\\
&\qquad\left.+\nu_1^*(d_1(x,y)-D_1)+\nu_2^*(d_2(x,z)-D_2)\vphantom{\frac{P_{Y|X}(y|x)}{Q_{Y}(y)}}\right)\\
&=\sum_{x,y,z}P_{X}(x)P_{YZ|X}(yz|x)\log\frac{P_{YZ|X}(y,z|x)}{Q_{YZ}(y,z)f(x,y,z|Q_{YZ},\lambda^*,\nu_1^*,\nu_2^*)}\\
&\geq \sum_{x}P_{X}(x)\left(\sum_{y,z}P_{YZ|X}(yz|x)\right)\log\frac{\sum_{y,z}P_{YZ|X}(yz|x)}{\sum_{y,z}Q_{YZ}(y,z)f(x,y,z|Q_{YZ},\lambda^*,\nu_1^*,\nu_2^*)}\\*
&=\sum_{x}P_{X}(x)\Lambda(x|Q_{YZ},\lambda^*,\nu_1^*,\nu_2^*).
\end{align}
\end{proof}

Recall that $P_{YZ|X}^*$ is the optimal test channel achieving $\rvR(R_1^*,D_1,D_2|P_X)$ and $P_{YZ}^*$ is the induced marginal distributions. Note that
\begin{align}
\rvR(R_1^*,D_1,D_2|P_X)
&=\inf_{Q_{YZ}}\inf_{P_{YZ|X}}F(P_{YZ|X},Q_{YZ}|R_1^*,D_1,D_2)\\
&\leq \inf_{P_{YZ|X}}F(P_{YZ|X},P_{YZ}^*|R_1^*,D_1,D_2)\\
&\leq F(P_{YZ|X}^*,P_{YZ}^*|R_1^*,D_1,D_2)\\
&=\rvR(R_1^*,D_1,D_2|P_X),\label{convexopt2},
\end{align}
where~\eqref{convexopt2} follows from~\eqref{convexopt}. Hence, $P_{YZ|X}^*$ achieves $\inf_{P_{YZ|X}}F(P_{YZ|X},P_{YZ}^*,P_{Z}^*|R_1^*,D_1,D_2)$. Invoking Lemma~\ref{relatefl}, we obtain
\begin{align}
\rvR(R_1^*,D_1,D_2|P_X)=\mathbb{E}_{P_X}\left[\Lambda(X|P_{YZ}^*,\lambda^*,\nu_1^*,\nu_2^*)\right],
\end{align}
and
\begin{align}
\nn&\Lambda(x|P_{YZ}^*,\lambda^*,\nu_1^*,\nu_2^*)\\*
&=\log\frac{P_{YZ|X}^*(y,z|x)}{P_{YZ}^*(y,z)}+\lambda^*\left(\log\frac{P_{Y|X}^*(y|x)}{P_Y^*(y)}-R_1^*\right)+\nu_1^*(d_1(x,y)-D_1)+\nu_2^*(d_2(x,z)-D_2).
\end{align}
The proof is complete by noting that
\begin{align}
\jmath_{YZ}(x,R_1^*,D_1,D_2|P_X)=\Lambda(x|P_{YZ}^*,P_{Y}^*,\lambda^*,\nu_1^*,\nu_2^*).
\end{align}

\subsection{Proof of Lemma~\ref{derivativer2}}
\label{proofderivativer2}
From the assumption in Lemma~\ref{derivativer2}, we obtain that $Q_X$ is supported on $\calX$. Let $Q_{YZ|X}^*$ be the optimal test channel achieving $\rvR(R_1^*,D_1,D_2|Q_X)$. Let $Q_{YZ}^*,Q_Y^*,Q_{Y|X}^*$ be the induced marginal distributions. Invoking Lemma~\ref{propertytilted}, we obtain
\begin{align}
\rvR(R_1^*,D_1,D_2|Q_X)&=\sum_{x}Q_X(x)\jmath_{YZ}(x,R_1^*,D_1,D_2|Q_X),
\end{align}
and
\begin{align}
\nn&\jmath_{YZ}(x,R_1^*,D_1,D_2|Q_X)\\
&=\log\frac{Q_{YZ|X}^*(y,z|x)}{Q_{YZ}^*(y,z)}+\lambda^*_{Q_X}\left(\log\frac{Q_{Y|X}^*(y|x)}{Q_Y^*(y)}-R_1^*\right)-\nu_{1,Q_X}^*(d_1(x,y)-D_1)-\nu_{2,Q_X}^*(d_2(x,z)-D_2),
\end{align}
where $\lambda^*_{Q_X},\nu_{1,Q_X}^*,\nu_{2,Q_X}^*$ are defined similarly as $\lambda^*,\nu_1^*,\nu_2^*$.
Hence,
\begin{align}
\frac{\partial \rvR(R_1^*,D_1,D_2|Q_X)}{Q_X(a)}\bigg|_{Q_X=P_X}
&=\jmath_{YZ}(a,R_1^*,D_1,D_2|P_X)+\frac{\partial }{\partial Q_X(a)}\mathbb{E}_{P_X}[\jmath_{YZ}(X,R_1^*,D_1,D_2|Q_X)].
\end{align}
Recall that given optimal test channel $P_{YZ|X}^*$ for $\rvR(R_1^*,D_1,D_2|P_X)$,
\begin{align}
\mathbb{E}_{P_X\times P_{YZ|X}^*}[d_1(X,Y)-D_1]&=0, \\
\mathbb{E}_{P_X\times P_{YZ|X}^*}[d_2(X,Z)-D_2]&=0.
\end{align}
Hence, we obtain for any $a\in\calX$, 
\begin{align}
\nn&\frac{\partial }{\partial Q_X(a)}\mathbb{E}_{P_X}[\jmath_{YZ}(X,R_1^*,D_1,D_2|Q_X)]\bigg|_{Q_X=P_X}\\
\nn&=\sum_{x}P_{X}(x)\sum_{y,z}P_{YZ|X}^*(yz|x)\left(\frac{\partial }{\partial Q_X(a)}\left(\log\frac{Q_{YZ|X}^*(y,z|x)}{Q_{YZ}^*(y,z)}\right)\bigg|_{Q_X=P_X}+\lambda^*\frac{\partial }{\partial Q_X(a)}\left(\log\frac{Q_{Y|X}^*(y|x)}{Q_Y^*(y)}-R_1^*\right)\bigg|_{Q_X=P_X}\right.\\*
&\qquad\left.+\frac{\partial \lambda^*_{Q_X} }{\partial Q_X(a)}\bigg|_{Q_X=P_X}\left(\log\frac{P_{Y|X}^*(y|x)}{P_Y^*(y)}-R_1^*\right)\right)\\*
&=-(1+\lambda^*),\label{notallzero}
\end{align}
where~\eqref{notallzero} follows for two reasons:
\begin{itemize}
\item The constant term $-(1+\lambda^*)$ comes from the first two terms in~\eqref{notallzero} which follows in a similar manner as~\cite[Theorem 2.2]{kostina2013lossy};
\item For optimal test channel, we have
\begin{align}
\mathbb{E}\left[\log \frac{P_{Y|X}^*(Y|X)}{P_{Y}(Y)}\right]=R_1^*.
\end{align}
\end{itemize}

\subsection{Proof of Lemma~\ref{tiltedconverse}}
\label{prooftiltedconverse}
Note that for any $(y,z)\in\calY\times\calZ$, from~\cite[Property 3]{kostina2012converse}, we have
\begin{align}
\mathbb{E}_{P_X}[\exp(\jmath_Y(X,D_1)+s_1^*(D-d_1(X,y)))]\leq 1,\\
\mathbb{E}_{P_X}[\exp(\jmath_Z(X,D_2)+s_2^*(D-d_2(X,z)))]\leq 1\label{kostinap3}.
\end{align}

Consider random transformations for encoders and decoders. Let random variables $U$ take values in $\{1,2,\ldots,M_1\}$ and $V$ take values in $\{1,2,\ldots,M_2\}$. Let $Q_{U}$ and $Q_V$ be uniform distributions on $\{1,2,\ldots,M_1\}$ and $\{1,2,\ldots,M_2\}$ respectively. We use $P_{U|X}$ and $P_{Y|U}$ to denote encoder $f_1$ and decoder $\phi_1$. Similarly, we use $P_{V|X}$ and $P_{Z|UV}$ to denote $f_2$ and $\phi_2$. Let $Q_Z$ be induced by $P_{Z|UV}$, $Q_U$ and $Q_V$. For any $\gamma_1\geq 0$ and $\gamma_2\geq 0$, we obtain
\begin{align}
&\nn\Pr\left(\jmath_Y(X,D_1|P_X)\geq \log M_1+\gamma_1~\mathrm{or}~\jmath_Z(X,D_2|P_X)\geq \log(M_1M_2)+\gamma_2\right)\\
\nn&\leq \Pr\left(\jmath_Y(X,D_1|P_X)\geq \log M_1+\gamma_1~\mathrm{or}~\jmath_Z(X,D_2|P_X)\geq \log(M_1M_2)+\gamma_2,~d_1(X,Y)\leq D_1,~d_2(X,Z)\leq D_2\right)\\
&\qquad+\epsilon_n(D_1,D_2)\\
\nn&\leq \Pr\left(\jmath_Y(X,D_1|P_X)\geq \log M_1+\gamma_1,~d_1(X,Y)\leq D_1\right)\\
&\qquad+\Pr\left(\jmath_Z(X,D_2|P_X)\geq \log(M_1M_2)+\gamma_2,~d_2(X,Z)\leq D_2\right)+\epsilon_n(D_1,D_2),\label{tiltupp}
\end{align}
where the first term in~\eqref{tiltupp} can be upper bounded by $\exp(-\gamma_1)$ as~\cite[Theorem 1]{kostina2012converse}. Here we upper bound the second term in~\eqref{tiltupp} as follows:
\begin{align}
&\nn\Pr\left(\jmath_Z(X,D_2|P_X)\geq \log(M_1M_2)+\gamma_2,~d_2(X,Z)\leq D_2\right)\\
&=\Pr\left(M_1M_2\leq \exp\left(\jmath_Z(X,D_2|P_X)\right)-\gamma_2,d_2(X,Z)\leq D_2\right)\\
&\leq \Pr\left(M_1M_2\leq \exp\left(\jmath_X(X,D_2|P_X)-\gamma_2\right)1\{d_2(X,Z)\leq D_2\}\right)\\
&\leq \frac{\mathbb{E}\left[\exp\left(\jmath_Z(X,D_2|P_X)-\gamma_2\right)1\{d_2(X,Z)\leq D_2\}\right]}{M_1M_2}\label{markovineq}\\
&\leq \frac{\exp(-\gamma_2)}{M_1M_2}\mathbb{E}\left[\exp\left(\jmath_Z(X,D_2|P_X)1\{d_2(X,Z)\leq D_2\}\right)\right]\\
&\leq \frac{\exp(-\gamma_2)}{M_1M_2}\mathbb{E}\left[\exp\left(\jmath_Z(X,D_2|P_X)+s_2^*(D_2-d_2(X,Z))\right)\right]\\
&=\frac{\exp(-\gamma_2)}{M_1M_2}\sum_{x}P_{X}(x)\sum_{u,v}P_{U|X}(u|x)P_{V|X}(v|x)\sum_{z}P_{Z|UV}(z|uv)\exp\left(\jmath_Z(x,D_2|P_X)+s_2^*(D_2-d_2(x,z))\right)\\
&\leq \exp(-\gamma_2)\sum_{x}P_X(x)\sum_{z}Q_{Z}(z)\exp\left(\jmath_Z(x,D_2|P_X)+s_2^*(D_2-d_2(x,z))\right)\label{tiltupp2}\\
&\leq \exp(-\gamma_2)\sum_{z}Q_{Z}(z)\mathbb{E}_{P_X}[\exp(\jmath_Z(X,D_2)+s_2^*(D-d_2(X,z)))]\\
&\leq \exp(-\gamma_2),\label{tiltupp3}
\end{align}
where~\eqref{markovineq} follows from Markov inequality;~\eqref{tiltupp2} follows from $P_{U|X}(u|x)\leq 1$, $P_{V|X}(v|x)\leq 1$ and the definition of $Q_Z$;~\eqref{tiltupp3} follows from~\eqref{kostinap3}.

\subsection{Proof of Lemma~\ref{uppexcess}}
\label{proofuppexcess}
Set $(R_1,R_2)=(R_{1,n},R_{2,n})$. Consider the following coding scheme. Given a source sequence $x^n$, encoder $f_1$ calculates the type $\hatT_{x^n}$ and sends it to both decoders with at most $|\calX|\log(n+1)$ nats. Then encoder $f_1$ calculates $R_Y(\hat{T}_{x^n},D_1)$ and $\rvR(R_{1,n},D_1,D_2|\hat{T}_{x^n})$, if $nR_Y(\hat{T}_{x^n},D_1)+(c_1+|\calX|)\log n>\log M_1$ or if $n\rvR(R_{1,n},D_1,D_2|\hat{T}_{x^n})+c_2\log n>\log(M_1M_2)$, the system declares an error directly. Otherwise, the two encoders operate as follows. Encoder $f_1$ chooses a set $\calB_{Y}$ specified by Lemma~\ref{typecovering} and sends out the codeword $y^n$ if $y^n=\argmin_{\tilde{y}^n} d_1(x^n,\tilde{y}^n)$. Then for each $y^n\in\calB_{Y}$, encoder $f_2$ chooses the set $\calB_{Z}(y^n)$ specified by Lemma~\ref{typecovering} and sends out the codeword $z^n$ if $z^n=\argmin_{\tilde{z}^n} d_2(x^n,z^n)$. At the decoder side, no error will be made. Hence, we have proved the upper bound on $\epsilon_n(D_1,D_2)$ in Lemma~\ref{uppexcess}.

\subsection{Proof of Lemma~\ref{typestrongconverse}}
\label{prooftypesc}
Define the set
\begin{align}
\calD_{Q_X}:=\left\{x^n\in\calT_{Q_X}:d_1(x^n,\phi_1(f_1(x^n)))\leq D_1,~d_2(x^n,\phi_2(f_1(x^n),f_2(x^n)))\leq D_2\right\}.
\end{align}
Recall that $U_{\calT_{Q_X}}$ denotes the uniform distribution over the type class $\calT_{Q_X}$. Let $\beta=\frac{\log n}{n}$. Define another distribution $Q_{\calT_{Q_X}}(x^n)$ such that
\begin{align}
Q_{\calT_{Q_{X}}}(x^n):=\frac{\exp(n(\alpha+\beta))U_{\calT_{Q_{X}}}(x^n)}{\exp(n(\alpha+\beta))U_{\calT_{Q_{X}}}(D_{Q_{X}})+(1-U_{\calT_{Q_{X}}}(D_{Q_{X}}))}
\end{align}
for $x^n\in\calD_{Q_{X}}$ and
\begin{align}
Q_{\calT_{Q_{X}}}(x^n):=\frac{U_{\calT_{Q_{X}}}(x^n)}{\exp(n(\alpha+\beta))U_{\calT_{Q_{X}}}(D_{Q_{X}})+(1-U_{\calT_{Q_{X}}}(D_{Q_{X}}))}
\end{align}
for $x^n\notin\calD_{Q_{X}}$.

From the assumption of the Lemma in \eqref{eqn:lb_type_str_conv}, we know that the $(n,M_1,M_2)$-code satisfies
\begin{align}
U_{\calT_{Q_X}}(\calD_X)\geq \exp(-n\alpha).
\end{align}
Hence, we obtain
\begin{align}
Q_{\calT_{Q_{X}}}(D_{Q_{X}})
&= \frac{\exp(n(\alpha+\beta))U_{\calT_{Q_{X}}}(D_{Q_{X}})}{\exp(n(\alpha+\beta))U_{\calT_{Q_{X}}}(D_{Q_{X}})+(1-U_{\calT_{Q_{X}}}(D_{Q_{X}}))}\\
&=\frac{\exp(n\beta)}{\exp(n\beta)+\exp(-n\alpha)\frac{1-U_{\calT_{Q_{X}}}(D_{Q_{X}})}{U_{\calT_{Q_{X}}}(D_{Q_{X}})}}\\
&\geq \frac{\exp(n\beta)}{\exp(n\beta)+1}\\
&\geq 1-\frac{1}{n},\label{betavalue}
\end{align}
where \eqref{betavalue} results from that $n\beta=\log n$.
Therefore, we have
\begin{align}
\mathbb{E}[d_1(X^n,Y^n)]
&=\sum_{x^n}Q_{\calT_{Q_X}}(x^n)d_1(x^n,\phi_1(f_1(x^n)))\\
&=\sum_{x^n\in\calD_{Q_X}}Q_{\calT_{Q_X}}(x^n)d_1(x^n,\phi_1(f_1(x^n)))+\sum_{x^n\notin\calD_{Q_X}}Q_{\calT_{Q_X}}(x^n)d_1(x^n,\phi_1(f_1(x^n)))\\
&\leq D_1+\frac{\overline{d}_1}{n},\label{distortion1}
\end{align}
and similarly
\begin{align}
\mathbb{E}[d_2(X^n,Z^n)]&\leq D_2+\frac{\overline{d}_2}{n}\label{distortion2}.
\end{align}
Let $J$ be the uniform random variable on $\{1,2,\ldots,n\}$ independent of all other random variables. By~\eqref{distortion1} and~\eqref{distortion2}, we obtain
\begin{align}
D_1+\frac{\overline{d}_1}{n}
&\geq \mathbb{E}[d_1(X^n,Y^n)]\\*
&=\mathbb{E}\left[\frac{1}{n}\sum_{i=1}^nd_1(X_i,Y_i)\right]\\*
&=\mathbb{E}[d_1(X_J,Y_J)].
\end{align}
and 
\begin{align}
D_2+\frac{\overline{d}_2}{n}
&\geq \mathbb{E}[d_2(X_J,Z_J)].
\end{align}

Now we apply weak converse argument here. Note that $S_1=f_1(X^n)$ and $Y^n=\phi_1(S_1)$. Hence, $X^n\to S_1 \to\hat{X}^n$. However, since $X^n\sim Q_{\calT_{Q_X}}$, $X^n$ is not i.i.d. Following a similar manner as converse proof in~\cite[pp.~59]{el2011network}, we obtain
\begin{align}
\log M_1
&\geq H(S_1)\\
&=I(X^n;S_1)\\
&\geq I(X^n;Y^n)\\
&=\sum_{i=1}^n I(X_i;Y^n|X^{i-1})\\
&\geq \sum_{i=1}^n I(X_i;Y_i|X^{i-1})\\
&\geq \sum_{i=1}^n \left(I(X_i;Y_i,X^{i-1})-I(X_i;X^{i-1})\right)\\
&\geq \sum_{i=1}^n I(X_i;Y_i)-\left(\sum_{i=1}^nH(X_i)-H(X^n)\right)\\
&=n I(X_J;Y_J|J)-\left(nH(X_J|J)-H(X^n)\right)\\
&=n I(X_J;Y_J|J)+nI(J;X_J)-\left(nH(X_J)-H(X^n)\right)\\
&=n I(X_J;Y_J,J)-\left(nH(X_J)-H(X^n)\right).
\end{align}
Note that $S_2=f_2(X^n)$ and $Z^n=\phi_2(S_1,S_2)$. Hence, in a similar manner, we obtain
\begin{align}
\log(M_1M_2)
&\geq H(S_1,S_2)\\
&=I(S_1,S_2;X^n)\\
&\geq I(X^n;Y^n,Z^n)\\
&\geq \sum_{i=1}^n I(X_i;Y_i,Z_i|X^{i-1})\\*
&=n I(X_J;Y_J,Z_J,J)-\left(nH(X_J)-H(X^n)\right).
\end{align}
Then in a similar manner as (152)--(154) in~\cite{watanabe2015second}, we have that there exists a conditional distribution $P_{YZ|X_J}$ such that
\begin{align}
\mathbb{E}[d_1(X_J,Y_J)]&=\mathbb{E}[d_1(X_J,Y)],\\
\mathbb{E}[d_2(X_J,Z_J)]&=\mathbb{E}[d_2(X_J,Z)],\\
I(X_J;Y_J,J)&=I(X_J;Y)=I(P_{X_J},P_{Y|X_J}),\\
I(X_J;Y_J,Z_J,J)&=I(X_J;YZ)=I(P_{X_J},P_{YZ|X_J}).
\end{align}
Then, in a similar manner as (155)--(157) in~\cite{watanabe2015second}, we can prove that $P_{X_J}(x)=Q_{X}(x)$. Hence, we conclude
\begin{align}
\mathbb{E}[d_1(X_J,Y_J)]&=\mathbb{E}_{Q_X\times Q_{YZ|X}}[d_1(X,Y)],\\
\mathbb{E}[d_2(X_J,Z_J)]&=\mathbb{E}_{Q_X\times Q_{YZ|X}}[d_2(X,Z)],\\
I(X_J;Y_J,J)&=I(Q_{X},P_{Y|X}),\\
I(X_J;Y_J,Z_J,J)&=I(Q_{X},P_{YZ|X}),
\end{align}
and
\begin{align}
H(X_J)=H(Q_X).
\end{align}
Following similar steps as (162)--(167) in~\cite{watanabe2015second}, we can prove
\begin{align}
\left|H(X_J)-\frac{1}{n}H(X^n)\right|&\leq \frac{|\calX|\log(n+1)}{n}+(\alpha+\beta).
\end{align} 
The proof is now complete by noting that $\beta=\frac{\log n}{n}$.

\subsection{Proof of Lemma~\ref{lbexcessp}}
\label{prooflbexcessp}
Set $\alpha=\log n/n$.
Given $(R_{1,n},R_{2,n})$, invoking Lemma~\ref{typestrongconverse}, we obtain that if $(R_{1,n},R_{2,n})\notin \rvR(D_{1,n},D_{2,n}|Q_{X})$, then
\begin{align}
\Pr\left(d_1(X^n,\hat{X}^n)>D_1~\mathrm{or}~d_2(X^n,Z^n)>D_2|X^n\in\calT_{Q_X}\right)
\geq 1-\frac{1}{n}.
\end{align}
Hence,
\begin{align}
\epsilon_n(D_1,D_2)
&=\sum_{Q_{X}\in\calP_n(\calX)}P_{X}^n(\calT_{Q_X})\Pr\left(d_1(X^n,\hat{X}^n)>D_1~\mathrm{or}~d_2(X^n,Z^n)>D_2|X^n\in\calT_{Q_X}\right)\\
&\geq \sum_{\substack{Q_X\in\calP_n(\calX):R_{1,n}<R_Y(Q_X,D_{1,n})\\R_{2,n}<\rvR(R_{1,n},D_{1,n},D_{2,n}|Q_X)}}P_{X}^n(\calT_{Q_X})\Pr\left(d_1(X^n,\hat{X}^n)>D_1~\mathrm{or}~d_2(X^n,Z^n)>D_2|X^n\in\calT_{Q_X}\right)\\
&\geq \sum_{\substack{Q_X\in\calP_n(\calX):R_{1,n}<R_Y(Q_X,D_{1,n})\\R_{2,n}<\rvR(R_{1,n},D_{1,n},D_{2,n}|Q_X)}}P_{X}^n(\calT_{Q_X})\left(1-\frac{1}{n}\right)\\
&\geq \Pr\left(R_{1,n}<R_Y(\hat{T}_{X^n},D_{1,n})~\mathrm{or}~R_{2,n}<\rvR(R_{1,n},D_{1,n},D_{2,n}|\hat{T}_{X^n})\right)-\frac{1}{n}.
\end{align}

\subsection{Proof of Lemma~\ref{gauuppexcess}}
\label{proofgauuppexcess}
Given $x^n$, if $x^n\notin \calU^{\xi}$, the system declares an error. Otherwise, encoder $f_1$ sends the type of $x^n$ by using no more than $\log k$ nats since there are $k$ different types $[1:k]$. Suppose $x^n\in\calU_i$. Encoder $f_1$ calculates $\log \frac{\Lambda(i)}{D_j},~j=1,2$. If $\log M_1<\frac{n}{2}\log \frac{\Lambda(i)}{D_1}+\frac{5}{2}\log n+\log k+\log 6$ or $\log (M_1M_2)<\frac{n}{2}\log \frac{\Lambda(i)}{D_2}+5\log n+2\log 6$, the system declares an error. Otherwise, invoking Lemma~\ref{Gaucovering}, we conclude that no error will be made. Define $\gamma_n=4\exp\left(-n I(\xi)\right)$. Hence,
\begin{align}
\epsilon_n(D_1,D_2)
&=\Pr(X^n\notin \calU^{\xi})+\sum_{i=1}^k \Pr(X^n\in\calU_i)\Pr\left(d_1(X^n,Y^n)>D_1~\mathrm{or}~d_2(X^n,Z^n)>D_2\bigg|X^n\in\calU_i\right)\\
&\leq \gamma_n+ \sum_{i=1}^k \Pr(X^n\in\calU_i)\Pr\left(R_{1,n}<\frac{1}{2}\log\frac{\Lambda(i)}{D_1}~\mathrm{or}~R_{2,n}<\frac{1}{2}\log\frac{\Lambda(i)}{D_2}\bigg|X^n\in\calU_i\right)\label{atypical}\\
&=\gamma_n+\sum_{i=1}^k \Pr(X^n\in\calU_i)\Pr\left(\Lambda(i)>D_1\exp(2R_{1,n})~\mathrm{or}~\Lambda(i)>D_2\exp(2R_{2,n})\bigg|X^n\in\calU_i\right)\\
&\leq \gamma_n+\sum_{i=1}^k \Pr\left(\frac{\|X^n\|^2}{n}+\delta\sigma^2>D_1\exp(2R_{1,n})~\mathrm{or}~\frac{\|X^n\|^2}{n}+\delta\sigma^2>D_2\exp(2R_{1,n}),~X^n\in\calU_i\right)\label{eqn:lambdabyx}\\
&\leq \gamma_n+\Pr\left(\frac{1}{n}\sum_{i=1}^n \frac{X_i^2}{\sigma^2}>\frac{D_1}{\sigma^2}\exp(2R_{1,n})-\delta~\mathrm{or}~\frac{1}{n}\sum_{i=1}^n \frac{X_i^2}{\sigma^2}>\frac{D_2}{\sigma^2}\exp(2R_{2,n})-\delta,~X^n\in\calU^{\xi}\right)\label{disjointset}\\
&\leq \gamma_n+\Pr\left(\frac{1}{n}\sum_{i=1}^n \frac{X_i^2}{\sigma^2}>\frac{D_1}{\sigma^2}\exp(2R_{1,n})-\delta~\mathrm{or}~\frac{1}{n}\sum_{i=1}^n \frac{X_i^2}{\sigma^2}>\frac{D_2}{\sigma^2}\exp(2R_{2,n})-\delta\right),
\end{align}
where~\eqref{atypical} follows from~\eqref{eqn:cramer}; \eqref{eqn:lambdabyx} follows because for $X^n\in\calU_i$ ($\calU_i$ was defined in \eqref{defui}), $\|X^n\|^2/n+\delta\sigma^2\geq \Lambda(i)$; \eqref{disjointset} follows since $\calU_i$ and $\calU_j$ are disjoint for any $i\neq j$. The proof of Lemma~\ref{gauuppexcess} is now complete.

\bibliographystyle{IEEEtran}
\bibliography{IEEEfull_2}

\begin{thebibliography}{10}
\providecommand{\url}[1]{#1}
\csname url@samestyle\endcsname
\providecommand{\newblock}{\relax}
\providecommand{\bibinfo}[2]{#2}
\providecommand{\BIBentrySTDinterwordspacing}{\spaceskip=0pt\relax}
\providecommand{\BIBentryALTinterwordstretchfactor}{4}
\providecommand{\BIBentryALTinterwordspacing}{\spaceskip=\fontdimen2\font plus
\BIBentryALTinterwordstretchfactor\fontdimen3\font minus
  \fontdimen4\font\relax}
\providecommand{\BIBforeignlanguage}[2]{{%
\expandafter\ifx\csname l@#1\endcsname\relax
\typeout{** WARNING: IEEEtran.bst: No hyphenation pattern has been}%
\typeout{** loaded for the language `#1'. Using the pattern for}%
\typeout{** the default language instead.}%
\else
\language=\csname l@#1\endcsname
\fi
#2}}
\providecommand{\BIBdecl}{\relax}
\BIBdecl

\bibitem{zhou2016sr}
L.~Zhou, V.~Y.~F. Tan, and M.~Motani, ``Second-order coding region for the
  discrete successive refinement source coding problem,'' in \emph{IEEE ISIT},
  July 2016, pp. 2414--2418.

\bibitem{rimoldi1994}
B.~Rimoldi, ``Successive refinement of information: characterization of the
  achievable rates,'' \emph{IEEE Trans. Inf. Theory}, vol.~40, no.~1, pp.
  253--259, Jan 1994.

\bibitem{equitz1991successive}
W.~H. Equitz and T.~M. Cover, ``Successive refinement of information,''
  \emph{IEEE Trans. Inf. Theory}, vol.~37, no.~2, pp. 269--275, 1991.

\bibitem{tan2015asymptotic}
V.~Y.~F. Tan, \emph{Asymptotic estimates in information theory with
  non-vanishing error probabilities}.\hskip 1em plus 0.5em minus 0.4em\relax
  Foundations and Trends{\textregistered} in Communications and Information
  Theory, 2014, vol.~11, no. 1-2.

\bibitem{altugwagner2014}
Y.~Alt\u{u}g and A.~B. Wagner, ``Moderate deviations in channel coding,''
  \emph{IEEE Trans. Info. Theory}, vol.~60, no.~8, pp. 4417--4426, 2014.

\bibitem{no2015strong}
A.~No, A.~Ingber, and T.~Weissman, ``Strong successive refinability and
  rate-distortion-complexity tradeoff,'' \emph{IEEE Trans. Inf. Theory},
  vol.~62, no.~6, pp. 3618--3635, June 2016.

\bibitem{vincent2014dispersion}
V.~Y.~F. Tan and O.~Kosut, ``On the dispersions of three network information
  theory problems,'' \emph{IEEE Trans. Inf. Th.}, vol.~60, no.~2, pp. 881--903,
  Feb 2014.

\bibitem{le2015inter}
S.~Q. Le, V.~Y.~F. Tan, and M.~Motani, ``A case where interference does not
  affect the channel dispersion,'' \emph{IEEE Transactions on Information
  Theory}, vol.~61, no.~5, pp. 2439--2453, May 2015.

\bibitem{watanabe2015second}
S.~Watanabe, ``Second-order region for {G}ray-{W}yner network,'' \emph{arXiv
  preprint arXiv:1508.04227}, 2015.

\bibitem{kanlis1996error}
A.~Kanlis and P.~Narayan, ``Error exponents for successive refinement by
  partitioning,'' \emph{IEEE Trans. Inf. Theory}, vol.~42, no.~1, pp. 275--282,
  Jan 1996.

\bibitem{wei2009strong}
W.~Gu and M.~Effros, ``A strong converse for a collection of network source
  coding problems,'' in \emph{Proc. IEEE ISIT}, June 2009, pp. 2316--2320.

\bibitem{Ben03}
V.~Bentkus, ``{On the dependence of the Berry-Esseen bound on dimension},''
  \emph{J.\ Stat.\ Planning and Inference}, vol. 113, pp. 385--402, 2003.

\bibitem{dembo2009large}
A.~Dembo and O.~Zeitouni, \emph{Large deviations techniques and
  applications}.\hskip 1em plus 0.5em minus 0.4em\relax Springer Science \&
  Business Media, 2009, vol.~38.

\bibitem{kostina2012converse}
V.~Kostina and S.~Verd{\'u}, ``A new converse in rate-distortion theory,'' in
  \emph{CISS}, March 2012, pp. 1--6.

\bibitem{verger2005covering}
J.-L. Verger-Gaugry, ``Covering a ball with smaller equal balls in {$R^n$},''
  \emph{Discrete \& Computational Geometry}, vol.~33, no.~1, pp. 143--155,
  2005.

\bibitem{arikan1998guessing}
E.~Ar{\i}kan and N.~Merhav, ``Guessing subject to distortion,'' \emph{IEEE
  Trans. Inf. Theory}, vol.~44, no.~3, pp. 1041--1056, 1998.

\bibitem{kelly2012reliability}
B.~G. Kelly and A.~B. Wagner, ``Reliability in source coding with side
  information,'' \emph{IEEE Trans. Inf. Theory}, vol.~58, no.~8, pp.
  5086--5111, 2012.

\bibitem{scarlett2015}
J.~Scarlett, ``On the dispersions of the {G}el'fand-{P}insker channel and dirty
  paper coding,'' \emph{IEEE Trans. Inf. Theory}, vol.~61, no.~9, pp.
  4569--4586, Sept 2015.

\bibitem{scarlett2015a}
J.~Scarlett and V.~Y.~F. Tan, ``Second-order asymptotics for the {Gaussian MAC}
  with degraded message sets,'' \emph{IEEE Trans. Inf. Theory}, vol.~61,
  no.~12, pp. 6700--6718, December 2015.

\bibitem{effros1999}
M.~Effros, ``Distortion-rate bounds for fixed- and variable-rate
  multiresolution source codes,'' \emph{IEEE Trans. Inf. Theory}, vol.~45,
  no.~6, pp. 1887--1910, Sep 1999.

\bibitem{tuncel2003additive}
E.~Tuncel and K.~Rose, ``Additive successive refinement,'' \emph{IEEE Trans.
  Inf. Theory}, vol.~49, no.~8, pp. 1983--1991, Aug 2003.

\bibitem{tuncel2003}
------, ``Error exponents in scalable source coding,'' \emph{IEEE Trans. Inf.
  Theory}, vol.~49, no.~1, pp. 289--296, Jan 2003.

\bibitem{strassen1962asymptotische}
V.~Strassen, ``Asymptotische absch{\"a}tzungen in shannons
  informationstheorie,'' in \emph{Trans. Third Prague Conf. Information
  Theory}, 1962, pp. 689--723.

\bibitem{hayashi2008source}
M.~Hayashi, ``Second-order asymptotics in fixed-length source coding and
  intrinsic randomness,'' \emph{IEEE Trans. Inf. Theory}, vol.~54, no.~10, pp.
  4619--4637, Oct 2008.

\bibitem{tan2014dispersions}
V.~Y.~F. Tan and O.~Kosut, ``On the dispersions of three network information
  theory problems,'' \emph{IEEE Trans. Inf. Th.}, vol.~60, no.~2, pp. 881--903,
  2014.

\bibitem{nomura2014}
R.~Nomura and T.~S. Han, ``Second-order {Slepian-Wolf} coding theorems for
  non-mixed and mixed sources,'' \emph{IEEE Trans. Inf. Theory}, vol.~60,
  no.~9, pp. 5553--5572, Sept 2014.

\bibitem{ingber2011}
A.~Ingber and Y.~Kochman, ``The dispersion of lossy source coding,'' in
  \emph{Proc. IEEE DCC}, March 2011, pp. 53--62.

\bibitem{kostina2012fixed}
V.~Kostina and S.~Verd{\'u}, ``Fixed-length lossy compression in the finite
  blocklength regime,'' \emph{IEEE Trans. Inf. Theory}, vol.~58, no.~6, pp.
  3309--3338, 2012.

\bibitem{watanabe2015}
S.~Watanabe, S.~Kuzuoka, and V.~Y.~F. Tan, ``Nonasymptotic and second-order
  achievability bounds for coding with side-information,'' \emph{IEEE Trans.
  Inf. Theory}, vol.~61, no.~4, pp. 1574--1605, April 2015.

\bibitem{yassaee2013technique}
M.~H. Yassaee, M.~R. Aref, and A.~Gohari, ``A technique for deriving one-shot
  achievability results in network information theory,'' in \emph{Proc. IEEE
  ISIT}, 2013, pp. 1287--1291.

\bibitem{kumagai2014}
W.~Kumagai and M.~Hayashi, ``Random number conversion via restricted storage,''
  in \emph{IEEE ISIT}, June 2014, pp. 2047--2051.

\bibitem{chen2007redundancy}
J.~Chen, D.-K. He, A.~Jagmohan, and L.~A. Lastras-Montano, ``On the
  redundancy-error tradeoff in {Slepian-Wolf} coding and channel coding,'' in
  \emph{Proc. IEEE ISIT}, 2007.

\bibitem{he2009redundancy}
D.-K. He, L.~A. Lastras-Monta{\v{n}}o, E.-H. Yang, A.~Jagmohan, and J.~Chen,
  ``On the redundancy of {Slepian--Wolf} coding,'' \emph{IEEE Trans. Inf.
  Theory}, vol.~55, no.~12, pp. 5607--5627, 2009.

\bibitem{polyanskiy2010channel}
Y.~Polyanskiy and S.~Verd\'{u}, ``Channel dispersion and moderate deviations
  limits for memoryless channels,'' in \emph{Proc. 48th Annu. Allerton Conf.},
  Sept. 2010, pp. 1334--1339.

\bibitem{altug2010moderate}
Y.~Alt\u{u}g and A.~B. Wagner, ``Moderate deviation analysis of channel coding:
  Discrete memoryless case,'' in \emph{Proc. IEEE ISIT}, Jun. 2010, pp.
  265--269.

\bibitem{altug2013lossless}
Y.~Alt\u{u}g, A.~B. Wagner, and I.~Kontoyiannis, ``Lossless compression with
  moderate error probability,'' in \emph{Proc. IEEE ISIT}, Jul. 2013, pp.
  1744--1748.

\bibitem{tan2012moderate}
V.~Y.~F. Tan, ``Moderate-deviations of lossy source coding for discrete and
  gaussian sources,'' in \emph{Proc. IEEE ISIT}, Jul. 2012, pp. 920--924.

\bibitem{borade2008}
S.~Borade and L.~Zheng, ``Euclidean information theory,'' in \emph{IEEE IZS},
  March 2008, pp. 14--17.

\bibitem{koshelev1981estimation}
V.~Koshelev, ``Estimation of mean error for a discrete successive-approximation
  scheme,'' \emph{Probl. Pered. Informat.}, vol.~17, no.~3, pp. 20--33, 1981.

\bibitem{el2011network}
A.~El~Gamal and Y.-H. Kim, \emph{Network information theory}.\hskip 1em plus
  0.5em minus 0.4em\relax Cambridge University Press, 2011.

\bibitem{kostina2013lossy}
V.~Kostina, ``Lossy data compression: Non-asymptotic fundamental limits,''
  Ph.D. dissertation, Princeton University, 2013.

\bibitem{wang2011}
D.~Wang, A.~Ingber, and Y.~Kochman, ``The dispersion of joint source-channel
  coding,'' \emph{arXiv preprint arXiv:1109.6310}, 2011.

\bibitem{gerrish}
A.~M. Gerrish, ``Estimation of information rates,'' Ph.D. dissertation, Yale
  University, New Haven, CT, 1963.

\bibitem{cvx}
M.~Grant and S.~Boyd, ``{CVX}: Matlab software for disciplined convex
  programming, version 2.1,'' \url{http://cvxr.com/cvx}, Mar. 2014.

\bibitem{csiszar2011information}
I.~Csiszar and J.~K{\"o}rner, \emph{Information theory: coding theorems for
  discrete memoryless systems}.\hskip 1em plus 0.5em minus 0.4em\relax
  Cambridge University Press, 2011.

\bibitem{Marton74}
K.~Marton, ``Error exponent for source coding with a fidelity criterion,''
  \emph{IEEE Trans. Inf. Theory}, vol.~20, no.~2, pp. 197--199, 1974.

\bibitem{zhong2006type}
Y.~Zhong, F.~Alajaji, and L.~L. Campbell, ``A type covering lemma and the
  excess distortion exponent for coding memoryless {Laplacian} sources,'' in
  \emph{23rd Biennial Symposium on Communications}, 2006, pp. 100--103.

\bibitem{tan2014state}
M.~Tomamichel and V.~Y.~F. Tan, ``Second-order coding rates for channels with
  state,'' \emph{IEEE Trans. Inf. Theory}, vol.~60, no.~8, pp. 4427--4448, Aug
  2014.

\bibitem{weissman2003inequalities}
T.~Weissman, E.~Ordentlich, G.~Seroussi, S.~Verdu, and M.~J. Weinberger,
  ``Inequalities for the {L1} deviation of the empirical distribution,''
  \emph{Hewlett-Packard Labs, Tech. Rep}, 2003.

\bibitem{chowberger}
J.~Chow and T.~Berger, ``Failure of successive refinement for symmetric
  {Gaussian} mixtures,'' \emph{IEEE Trans. Inf. Theory}, vol.~43, no.~1, pp.
  350--352, Jan 1997.

\bibitem{fu2002rate}
F.-W. Fu and R.~W. Yeung, ``On the rate-distortion region for multiple
  descriptions,'' \emph{IEEE Trans. Inf. Theory}, vol.~48, no.~7, pp.
  2012--2021, 2002.

\end{thebibliography}
\end{document}